%% file: main.tex
\documentclass{article}

\usepackage{times}
\usepackage{latexsym}
\usepackage{amsmath}
\usepackage{amsfonts}
\usepackage{amsthm}
\usepackage{epsfig}
\usepackage{cite}
\usepackage{graphicx}
\usepackage[noend]{algorithmic}
\usepackage{algorithm}
\usepackage{fullpage}
\usepackage{txfonts} 
\usepackage{xcolor,colortbl}
\usepackage{cite}

\input{commands}

\title{Worst-case Optimal Join Algorithms}

\author{Hung Q. Ngo\\
Computer Science and Engineering,\\
SUNY Buffalo,\\
U.S.A.
\and
Ely Porat\\
Computer Science,\\
Bar-Ilan University,\\
Israel
\and
Christopher R\'e\\
Computer Science,\\
University of Wisconsin--Madison\\
U.S.A.
\and
Atri Rudra\\
Computer Science and Engineering,\\
SUNY Buffalo,\\
U.S.A.}

\begin{document}

\maketitle

\input{abstract}

\input{intro}
\input{prelim}
\input{main_results}

\input{extensions}
\input{conclusions}

\bibliographystyle{acm}
\bibliography{main}
  


\end{document}

%% file: commands.tex
\newcommand{\eat}[1]{}
\newcommand{\set}[1]{\ensuremath \{#1\}}

\newcommand{\univ}{\text{\sc univ}}
\newcommand{\lbl}{\text{\sc label}}
\newcommand{\lc}{\text{\sc lc}}
\newcommand{\rc}{\text{\sc rc}}
\newcommand{\nil}{\text{\sc nil}}
\newcommand{\poly}{\mathrm{poly}}
\newcommand{\ret}{\text{Ret}}

\newcommand{\eps}{\epsilon}

\newcommand{\T}{\mathcal{T}}
\newcommand{\R}{\mathbb{R}}

\newcommand{\la}{\leftarrow}

\newcommand{\argmin}{\mathop{\text{argmin}}}

\newcommand{\np}{\mathsf{NP}}
\newcommand{\rp}{\mathsf{RP}}

\newcommand{\be}{\begin{enumerate}}
\newcommand{\ee}{\end{enumerate}}
\newcommand{\bi}{\begin{itemize}}
\newcommand{\ei}{\end{itemize}}
\newcommand{\beq}{\begin{equation}}
\newcommand{\eeq}{\end{equation}}

\newcommand{\bp}{\begin{proof}}
\newcommand{\ep}{\end{proof}}
\newcommand{\bcor}{\begin{cor}}
\newcommand{\ecor}{\end{cor}}
\newcommand{\bthm}{\begin{thm}}
\newcommand{\ethm}{\end{thm}}
\newcommand{\blmm}{\begin{lmm}}
\newcommand{\elmm}{\end{lmm}}
\newcommand{\bdefn}{\begin{defn}}
\newcommand{\edefn}{\end{defn}}
\newcommand{\bprop}{\begin{prop}}
\newcommand{\eprop}{\end{prop}}
\newcommand{\bconj}{\begin{conj}}
\newcommand{\econj}{\end{conj}}
\newcommand{\bopm}{\begin{opm}}
\newcommand{\eopm}{\end{opm}}
\newcommand{\brmk}{\begin{rmk}}
\newcommand{\ermk}{\end{rmk}}

\newcommand{\suchthat}{\ | \ }

\newcommand{\mv}[1]{\mathbf{#1}}

\newtheorem{thm}{Theorem}[section]
\newtheorem{lmm}[thm]{Lemma}
\newtheorem{prop}[thm]{Proposition}
\newtheorem{cor}[thm]{Corollary}

\theoremstyle{definition}              

\newtheorem{opm}[thm]{Open Problem}
\newtheorem{conj}[thm]{Conjecture}
\newtheorem{example}[thm]{Example}

\newtheorem{defn}[thm]{Definition}

\newtheorem{rmk}[thm]{Remark}

\newcommand{\bbox}{
\begin{center}
\begin{tabular}{|c|}
\hline
}
\newcommand{\ebox}{
\\
\hline
\end{tabular}
\end{center}
}

\newlength{\toppush}
\def\subjnum{CSE 713}
\def\subjname{Probabilistically Checkable Proofs and Inapproximability}
\def\doheading#1#2#3{\vfill\eject\vspace*{-\toppush}%
  \vbox{\hbox to\textwidth{{\bf}
  \subjnum: \subjname
  \hfil Lecturer: Hung Q. Ngo}%
  \hbox to\textwidth{{\bf} SUNY at Buffalo, Fall 2004\hfil#3\strut}%
  \hrule}
}

\newcommand{\deter}[1]{C(#1)}
\newcommand{\ndeter}[1]{D(#1)}

%% file: abstract.tex
\begin{abstract}
\noindent
Efficient join processing is one of the most fundamental and
well-studied tasks in database research. In this work, we examine
algorithms for natural join queries over many relations and describe a novel
algorithm to process these queries optimally in terms of worst-case
data complexity. Our result builds on recent work by Atserias, Grohe,
and Marx, who gave bounds on the size of a full conjunctive query in
terms of the sizes of the individual relations in the body of the
query. These bounds, however, are not constructive: they rely on
Shearer's entropy inequality which is information-theoretic. 
Thus, the previous results leave open the question of whether there exist 
algorithms whose running time achieve these optimal bounds. An answer to this
question may be interesting to database practice, as it is known that
any algorithm based on the traditional select-project-join style
plans typically employed in an RDBMS are asymptotically slower than
the optimal for some queries. We construct an algorithm whose running
time is worst-case optimal for all natural join queries. 
Our result may be of independent interest, as our algorithm also yields a 
constructive proof of the general fractional cover bound by Atserias, Grohe, 
and Marx without using Shearer's inequality. This bound implies two famous
inequalities in geometry: the Loomis-Whitney inequality and the
Bollob\'as-Thomason inequality. Hence, our results algorithmically
prove these inequalities as well. Finally, we discuss how our algorithm can
be used to compute a relaxed notion of joins.
\end{abstract}

%% file: intro.tex
\section{Introduction}

Recently, Grohe and Marx~\cite{GM06} and Atserias, Grohe, and
Marx~\cite{AGM08} (AGM's results henceforth) derived tight bounds on
the number of output tuples of a {\em full conjunctive
  query}\footnote{A full conjunctive query is a conjunctive query
  where every variable in the body appears in the head.} in terms of
the sizes of the relations mentioned in the query's body. As query
output size estimation is fundamentally important for efficient query
processing, these results have generated a great deal of excitement.

To understand the spirit of AGM's results, consider the following
example where we have a schema with three attributes, $A$, $B$, and
$C$, and three relations, $R(A,B)$, $S(B,C)$ and $T(A,C)$, defined
over those attributes. Consider the following natural join
query:
\begin{equation}
q = R \Join S \Join T 
\label{eq:q}
\end{equation}
Let $q(I)$ denote the set of tuples that is output from applying $q$
to a database instance $I$, that is the set of triples of constants 
$(a,b,c)$ such that $R(ab)$, $S(bc)$, and $T(ac)$ are in $I$. 
Our goal is to bound the
number of tuples returned by $q$ on $I$, denoted by $|q(I)|$, in terms of
$|R|$, $|S|$, and $|T|$. For simplicity, let us consider
the case when $|R|=|S|=|T|=N$.
A straightforward bound is $|q(I)| \leq N^3$. 
One can obtain a better bound by noticing that any pair-wise join 
(say $R\Join S$) will contain $q(I)$ in it as $R$ and $S$ together contain all 
attributes (or they ``cover'' all the attributes). 
This leads to the bound $|q(I)|\le N^2$. 
AGM showed that one can get a better upper bound of 
$|q(I)|\le N^{3/2}$ by generalizing the notion of cover to a so-called
``fractional cover'' (see Section~\ref{sec:notations}). 
Moreover, this estimate is tight in the sense that
for infinitely many values of $N$, one can find a database instance
$I$ that for which $|q(I)| = N^{3/2}$. These non-trivial
estimates are exciting to database researchers as they offer
previously unknown, nontrivial methods to estimate the cardinality of
a query result -- a fundamental problem to support efficient query
processing.

More generally, given an arbitrary natural-join query $q$ and given
the sizes of input relations, the AGM method can generate an upper
bound $U$ such that $|q(I)| \leq U$, where $U$ depends on the ``best''
fractional cover of the attributes.  This ``best'' fractional cover
can be computed by a linear program (see Section~\ref{sec:notations} for
more details).  Henceforth, we refer to this inequality as the {\em
  AGM's fractional cover inequality}, and the bound $U$ as the {\em
  AGM's fractional cover bound}.  They also show that the bound is
essentially optimal in the sense that for infinitely many sizes of
input relations, there exists an instance $I$ such that each relation
in $I$ is of the prescribed size and $|q(I)| = U$.

AGM's results leave open whether one can compute the actual set $q(I)$
in time $O(U)$. In fact, AGM observe this issue and presented an
algorithm that computes $q(I)$ with a running time of $O(|q|^2 \cdot
U\cdot N)$ where $N$ is the cardinality of the largest input relation
and $|q|$ denotes the size of the query $q$.  AGM establish that their
join-project plan can in some cases
be super-polynomially better than any join-only
plan.  However, AGM's join algorithm is not optimal.  Even on the
above example of \eqref{eq:q}, we can construct a family of database
instances $I_1, I_2, \dots, I_N, \dots, $ such that in the $N$th
instance $I_N$ we have $|R| = |S| = |T| = N$ and both AGM's algorithm
and any join-only plan take $\Omega(N^2)$-time even though from AGM's
bound we know that $|q(I)| \leq U = N^{3/2}$, which is the best
worst-case run-time one can hope for.

The $\sqrt N$-gap on a small example motivates our central question.
In what follows, {\em natural join queries} are defined as the 
join of a set of relations $R_1,\dots,R_m$. 

\noindent
\begin{quote}
  \textbf{Optimal Worst-case Join Evaluation Problem} (Optimal Join
  Problem).  {\em Given a fixed database schema $\bar R =
    \left\{R_i(\bar A_i)\right\}_{i=1}^{m}$ and an $m$-tuple of
    integers $\bar N = (N_1,\dots,N_m)$. Let $q$ be the natural join
    query joining the relations in $\bar R$ and let $I(\bar N)$ be the
    set of all instances such that $|R_i^{I}| = N_i$ for
    $i=1,\dots,m$. Define $U = \sup_{I \in I(\bar N)} |q(I)|$. Then, the
    optimal worst-case join evaluation problem is to evaluate $q$ in
    time $O(U + \sum_{i=1}^{m} N_i)$}.
\end{quote} 

\noindent
Since any algorithm to produce $q(I)$ requires time at least $|q(I)|$,
an algorithm that solves the above problem would have an optimal
worst-case data-complexity.\footnote{In an RDBMS, one computes information,
e.g., indexes, offline that may obviate the need to read the entire
input relations to produce the output. In a similar spirit, we can
extend our results to evaluate any query $q$ in time $O(U)$,
removing the term $\sum_i N_i$ by precomputing some indices.}
(Note that we are mainly concerned with data complexity and thus the 
$O(U)$ bound above ignores the dependence on $|q|$. Our results
have a small $O(|q|)$ factor.)

Implicitly, this problem has been studied for over three decades: a 
modern RDBMS use decades of highly tuned algorithms to efficiently produce 
query results. Nevertheless, as we described above, such systems are
asymptotically suboptimal -- even in the above simple example of
\eqref{eq:q}. Our main result is an algorithm that achieves
asymptotically optimal worst-case running times for all 
conjunctive join queries.

We begin by describing connections between AGM's inequality and a family of
inequalities in geometry. In particular, we show that the
AGM's inequality is {\em equivalent} to the discrete version of a geometric
inequality proved by Boll\'obas and Thomason 
(\cite{MR1338683}, Theorem 2). 
This equivalence is shown in Section~\ref{sec:bt:equiv}.

Our ideas for an algorithm solving the optimal join problem begin by
examining a special case of the Boll\'obas-Thomason (BT) inequality:
the classic Loomis-Whitney (LW) inequality~\cite{MR0031538}.  The LW
inequality bounds the measure of an $n$-dimensional set in terms of
the measures of its $(n-1)$-dimensional projections onto the
coordinate hyperplanes.  The query \eqref{eq:q} and its bound
$|q(I)|\leq \sqrt{|R||S||T|}$ is {\em exactly} the LW inequality with
$n=3$ applied to the discrete measure.  Our algorithmic development
begins with a slight generalization of the query $q$ in
\eqref{eq:q}. We describe an algorithm for join queries which have the
same format as in the LW inequality setup with $n\geq 3$.  In
particular, we consider ``LW instances'' of the optimal join problem,
where the query is to join $n$ relations whose attribute sets are all
the distinct $(n-1)$-subsets of a universe of $n$ attributes.  Since
the LW inequality is tight, and our join algorithm has running time
that is asymptotically data-optimal for this class of queries (e.g.,
$O(N^{3/2})$ in our motivating example), our algorithm is
data-complexity optimal in the worst case for LW instances.

Our algorithm for LW instances exhibits a key twist compared to a
conventional join algorithm. The twist is that the join algorithm
partitions the values of the join key on each side of the join into
two sets: those values that are {\em heavy} and those values that are
{\em light}. Intuitively, a value of a join key is heavy if its fanout
is high enough so that joining all such join keys could violate the
size bound (e.g., $N^{3/2}$ above). The art is selecting the precise
fanout threshold for when a join key is heavy. This per-tuple choice
of join strategy is not typically done in standard RDBMS join
processing. 

Building on our algorithm for LW instances, we next describe our main
result: an algorithm to solve the optimal join problem for all join
queries.  In particular, we design an algorithm for evaluating join
queries which not only {\em proves} AGM's fractional cover inequality
{\em without} using the information-theoretic Shearer's inequality,
but also has a running time that is linear in the bound (modulo
pre-processing time).  As AGM's inequality implies the BT and LW
inequalities, our result is the first algorithmic proof of these
geometric inequalities as well.  To do this, we must carefully select
which projections of relations to join and in which order our
algorithm joins relations on a ``per tuple'' basis as in the
LW-instance case.  Our algorithm computes these orderings, and then at
each stage it performs an algorithm that is similar to the algorithm
we used for LW instances.

Our example also shows that standard join algorithms are suboptimal,
the question is, {\it when do classical RDBMS algorithms have higher
worst-case run-time than our proposed approach?}  AGM's analysis of their
join-project algorithm leads to a worst case run-time complexity that is a
factor of the largest relation worse than the AGM's bound.  
To investigate whether AGM's analysis is tight or not, we ask a
sharper variant of this question: {\it Given a query $q$ does there
exist a family of instances $I$ such that our algorithm runs
asymptotically faster than a standard binary-join-based plan or
AGM's join-project plan?}  We give a partial answer to this question
by describing a sufficient syntactic condition for the query $q$ such
that for each $k \geq 2$, we can construct a family of instances where each
relation is of size $N$ such that any binary-join plan as well as AGM's
algorithm will need time 
$\Omega(N^{2}/k^2)$, while the fractional cover bound is 
$O(N^{1 + 1/(k-1)})$ -- an asymptotic gap. We then show through a more 
detailed analysis that our algorithm on these instances takes 
$O(k^2N)$-time.

We consider several extensions and improvements of our main result. In
terms of the dependence on query size, our algorithms are also
efficient (at most linear in $|q|$, which is better than the quadratic
dependence in AGM) for full queries, but they are not necessarily
optimal.  In particular, if each relation in the schema has arity $2$,
we are able to give an algorithm with better query complexity than our
general algorithm. This shows that in general our algorithm's
dependence on the factors of the query is not the best possible. We
also consider computing a relaxed notion of joins and give worst-case
optimal algorithms for this problem as well.

\paragraph*{Outline} The remainder of the paper is organized as
follows: in the rest of this section, we describe related work.
In Section~\ref{sec:notations} we describe our notation and formulate the 
main problem. 
Section~\ref{sec:bt:equiv} proves the connection
between AGM's inequality and BT inequality. 
In Section \ref{sec:LW-algo} we present a data-optimal join algorithm for 
LW instances, and then
extend this to arbitrary join queries in Section~\ref{sec:all:j}. 
We discuss the limits of performance of prior approaches and our 
approach in more detail in Section~\ref{sec:limits}. In
Section~\ref{sec:extensions}, we describe several extensions. We conclude in
Section~\ref{sec:conc}.

\subsection*{Related Work}

Grohe and Marx \cite{GM06} made the first (implicit) connection between 
fractional edge cover and the output size of a conjunctive query.
(Their results were stated for constraint satisfaction problems.)
Atserias, Grohe, and Marx \cite{AGM08} extended Grohe and Marx's results
in the database setting.

The first relevant AGM's result is the following inequality. 
Consider a join query over relations $R_e$, $e\in E$, where $E$ is a 
collection of subsets of an attribute ``universe'' $V$, and 
relation $R_e$ is on attribute set $e$. 
Then, the number of output tuples is bounded above by 
$\prod_{e\in E}|R_e|^{x_e}$, where $\mv x = (x_e)_{e\in E}$ is an
{\em arbitrary} fractional cover of the hypergraph $H=(V,E)$.

They also showed that this bound
is tight. In particular, for infinitely many positive integers $N$ there is
a database instance with $|R_e|=N$, $\forall e\in E$, and the 
upper bound gives the actual number of output tuples.
When the sizes $|R_e|$ were given as inputs to the (output size estimation)
problem, obviously the best upper bound is obtained by picking the fractional
cover $\mv x$ which minimizes the linear objective function
$\sum_{e\in E}(\log|R_e|)\cdot x_e$.
In this ``size constrained'' case, however, their lower bound is 
off from the upper bound by
a factor of $2^n$, where $n$ is the total number of attributes. 
AGM also presented an inapproximability result which justifies this gap. 
Note, however, that the gap is only dependent on the query size
and the bound is still asymptotically optimal in the 
data-complexity sense.

The second relevant result from AGM is a join-project plan with
running time $O\left(|q|^2 N_{\text{max}}^{1+ \sum x_e}\right)$,
where $N_{\text{max}}$ is the maximum size of input relations
and $|q| = |V|\cdot|E|$ is the query size.

The AGM's inequality contains as a special case the discrete versions
of two well-known inequalities in 
geometry: the \textit{Loomis-Whitney} (LW) 
inequality~\cite{MR0031538} and its generalization the 
\textit{Bollob\'as-Thomason} (BT) inequality~\cite{MR1338683}.
There are two typical proofs of the discrete LW and BT inequalities.
The first proof is by induction using H\"older's inequality~\cite{MR1338683}.
The second proof (see Lyons and Peres~\cite{Lyons-Peres})
essentially uses ``equivalent'' entropy inequalities
by Han \cite{MR0464499} and its generalization by Shearer \cite{MR859293},
which was also the route Grohe and Marx \cite{GM06} took to prove AGM's bound.
All of these proofs are non-constructive.

There are many applications of the discrete LW and BT inequalities.
The $n=3$ case of the LW inequality was used to prove communication
lower bounds for matrix multiplication on distributed memory parallel
computers \cite{DBLP:journals/jpdc/IronyTT04}. The inequality was used
to prove submultiplicativity inequalities regarding sums of sets of
integers~\cite{MR2676833}.  In~\cite{Lehman-Lehman}, a special case of
BT inequality was used to prove a network-coding bound. Recently, some
of the authors of this paper have used our \textit{algorithmic}
version of the LW inequality to design a new sub-linear time decodable
compressed sensing matrices~\cite{GNPRS} and efficient pattern
matching algorithms~\cite{NPR-pattern}.

Inspired by AGM's results, Gottlob, Lee, and Valiant~\cite{GLVV09}
provided bounds for conjunctive queries with functional
dependencies. For these bounds, they defined a new notion of
``coloring number'' which comes from the dual linear program of the
fractional cover linear program. This allowed them to generalize
previous results to all conjunctive queries, and to study several
problems related to tree-width.

Join processing algorithms are one of the most studied algorithms in
database research.  A staggering number of variants have been
considered, we list a few: Block-Nested loop join, Hash-Join, Grace,
Sort-merge (see Grafe~\cite{graefe93} for a survey). Conceptually, it
is interesting that none of the classical algorithms consider
performing a per-tuple cardinality estimation as our algorithm
does. It is interesting future work to implement our algorithm to
better understand its performance.

Related to the problem of estimating the size of an output is
cardinality estimation. A large number of structures have been
proposed for cardinality
estimation~\cite{Ioannidis03,DBLP:conf/sigmod/PoosalaIHS96,wavelet,DBLP:conf/pods/AlonGMS99,DBLP:conf/vldb/KonigW99,DBLP:conf/vldb/JagadishKMPSS98},
they have all focused on various sub-classes of queries and deriving
estimates for arbitrary query expressions has involved making
statistical assumptions such as the independence and containment
assumptions which result in large estimation
errors~\cite{ioannidischristodoulakis91}. Follow-up work has
considered sophisticated probability models, 
Entropy-based models~\cite{DBLP:conf/vldb/MarklMKTHS05,isomere:2006}
and graphical models~\cite{DBLP:journals/pvldb/TzoumasDJ11}. In contrast, 
in this work we examine the {\em worst case behavior} of algorithms in 
terms of its cardinality estimates. 
In the special case when the join graph is acyclic, there are several
known results which achieve (near) optimal run time with respect to the
output size \cite{DBLP:conf/pods/PaghP06, DBLP:journals/jcss/Willard96}.

On a technical level, the work {\em adaptive query processing} is
related, e.g., Eddies~\cite{DBLP:conf/sigmod/HellersteinA00} and
RIO~\cite{DBLP:conf/sigmod/BabuBD05}. The main idea is that to
compensate for bad statistics, the query plan may adaptively be
changed (as it better understands the properties of the data). While
both our method and the methods proposed here are adaptive in some
sense, our focus is different: this body of work focuses on heuristic
optimization methods, while our focus is on provable worst-case
running time bounds. A related idea has been considered in practice:
heuristics that split tuples based on their fanout have been deployed
in modern parallel databases to handle
skew~\cite{DBLP:conf/sigmod/XuKZC08}.  This idea was not used to
theoretically improve the running time of join algorithms. We are
excited by the fact that a key mechanism used by our algorithm has
been implemented in a modern commercial system.

%% file: prelim.tex
\newcommand{\D}{\mathbf{D}}

\section{Notation and Formal Problem Statement} 
\label{sec:notations}

We assume the existence of a set of attribute names $\mathcal{A} =
A_1,\dots, A_n$ with associated domains $\D_1,\dots,\D_n$ and infinite
set of relational symbols $R_1,R_2, \dots$. A relational schema for
the symbol $R_i$ of arity $k$ is a tuple $\bar A_i = (A_{i_1}, \dots,
A_{i_k})$ of distinct attributes that defines the attributes of the
relation. A relational database schema is a set of relational symbols
and associated schemas denoted by $R_1(\bar A_1), \dots, R_m(\bar
A_m)$. A relational instance for $R(A_{i_1},\dots,A_{i_k})$ is a
subset of $\D_{i_1} \times \dots \times \D_{i_k}$. A relational database $I$ is
an instance for each relational symbol in schema, denoted by
$R_i^{I}$. A {\em natural join} query (or simply query) $q$ is
specified by a finite subset of relational symbols $q \subseteq
\mathbb{N}$, denoted by $\Join_{i \in q} R_i$. Let $\bar A(q)$ denote
the set of all attributes that appear in some relation in $q$, that is
$\bar A(q) = \{ A \mid A \in \bar A_i \text{ for some } i \in q\}$.  
Given a tuple
$\mv t$ we will write $\mv t_{\bar A}$ to emphasize that its support
is the attribute set $\bar A$. Further, for any $\bar S\subset \bar A$
we let $\mv t_{\bar S}$ denote $\mv t$ restricted to $\bar S$.
Given a database instance $I$, the
output of the query $q$ on $I$ is denoted $q(I)$ and is defined as

\[ q(I) \stackrel{\mathrm{def}}{=} \left\{ \mv t \in \D^{\bar A(q)} \suchthat 
\mv t_{\bar A_{i}} \in R^{I}_i \text{ for each } i \in q\right\} \]
where $\D^{\bar A(q)}$ is a shorthand for $\times_{i : A_i \in \bar A(q)} \D_i$.

We also use the notion of a {\em semijoin}: Given two relations
$R(\bar A)$ and $S(\bar B)$ their semijoin $R \lJoin S$  is defined by
\[ R \lJoin S \stackrel{\mathrm{def}}{=} 
  \left\{ \mv t \in R : \exists \mv u \in S \text{ s.t. } 
\mv t_{\bar A \cap \bar B} = \mv u_{\bar A \cap \bar B} \right\}.
\]
For any relation $R(\bar A)$, and any subset $\bar S\subseteq \bar A$
of its attributes, let $\pi_{\bar S}(R)$ denote the {\em projection} of 
$R$ onto $\bar S$, i.e.
\[ \pi_{\bar S}(R) = \left\{\mv t_{\bar S} \suchthat
   \exists \mv t_{\bar A\setminus \bar S}, (\mv t_{\bar S}, \mv t_{\bar A\setminus \bar S})
   \in R \right\}.
\] 
For any tuple $\mv t_{\bar S}$,
define the {\em $\mv t_{\bar S}$-section} of $R$ as
\[ R[\mv t_{\bar S}] = \pi_{\bar A\setminus \bar S}(R \lJoin \set{\mv t_{\bar S}}).
\]

\paragraph*{From Join Queries to Hypergraphs}
A query $q$ on attributes $\bar A(q)$ can be viewed as a hypergraph
$H=(V,E)$ where $V = \bar A(q)$ and there is an edge $e_i =\bar A_i$
for each $i \in q$. Let $N_e = |R_e|$ be the number of tuples in
$R_e$. \textit{From now on we will use the hypergraph and the original notation for the query
interchangeably.}

We use this hypergraph to introduce the {\em fractional edge
  cover polytope} that plays a central role in our technical
developments. The fractional edge cover polytope defined by $H$
is the set of all points $\mv x =(x_e)_{e\in E}\in \R^E$ such that
\begin{eqnarray*}
\sum_{v \in e} x_e &\geq& 1, \ \text{for any $v \in V$}\\
x_e &\geq &0, \text{for any $e \in E$}
\end{eqnarray*}

Note that the solution $x_e = 1$ for $e \in E$ is always feasible
for hypergraphs representing join queries.
A point $\mv x$ in the polytope is also called a {\em fractional (edge)
cover solution} of the hypergraph $H$.

Atserias, Grohe, and Marx~\cite{AGM08} establish that, for {\em any} 
point $\mv x = (x_e)_{e\in E}$ in the fractional edge cover polytope
\begin{equation}
 |\Join_{e\in E} R_e| \leq \prod_{e\in E}N_e^{x_e}. 
\label{eqn:agm08-bound}
\end{equation}

The bound is proved nonconstructively using Shearer's entropy inequality
\cite{MR859293}. 
However, AGM provide an algorithm based on join-project plans that runs 
in time
$O(|q|^2\cdot N_{\max}^{1+\sum_{e} x_e})$ where $N_{\max} = \max_{e \in E} N_e$.
They observed that for a fixed hypergraph $H$ and given sizes $N_e$ 
the bound \eqref{eqn:agm08-bound} can be minimized by solving the linear program
which minimizes the linear objective $\sum_{e} (\log N_e)\cdot x_e$
over fractional edge cover solutions $\mv x$. 
(Since in linear time we can figure out if we have an empty relation, 
and hence an empty output), for the rest of the paper we are
always going to assume that $N_e\ge 1$.) 
Thus, the formal
problem that we consider recast in this language is:

\begin{defn}[OJ Problem -- Optimal Join Problem]
With the notation above, design an algorithm to compute $\Join_{e\in
  E} R_e$ with running time
\[ O\left( f(|V|,|E|) \cdot \left(\prod_{e\in E}N_e^{x_e} + \sum_{e\in
      E}N_e\right)\right). \] Here $f(|V|,|E|)$ is ideally a
polynomial with (small) constant degree, which only depends on the
query size.  The linear term $\sum_{e\in E}N_e$ is to read the input.
Hence, such an algorithm would be data-optimal in the worst case.\footnote{Following GLV~\cite{GLVV09}, we assume in this work
  that given relations $R$ and $S$ one can compute $R \Join S$ in time
  $O(|R| + |S| + |R \Join S|)$. This only holds in an amortized sense
  (using hashing). To acheive true worst case results, one can use
  sorting operations which results in a $\log$ factor increase in
  running time. }
\end{defn}

We recast our motivating example from the introduction in our
notation. Recall we are given, $R(A,B), S(B,C), T(A,C)$, so $V =
\set{A,B,C}$ and three edges corresponding each to $R$, $S$, and $T$,
which are $E=\set{\set{A,B}$, $\set{B,C}$, $\set{A,C}}$
respectively. Thus, $|V| = 3$ and $|E| = 3$. If we are given that $N_e
= N$, one can check that the optimal solution to the LP is $x_e =
\frac{1}{2}$ for $e \in E$ which has the objective value $\frac{3}{2}
\log N$; in turn, this gives $\sup_{I \in I(\bar N)} |q(I)| \leq
N^{3/2}$ (recall $I(\bar N) = \set{ I : |R_e^{I}| = N_e \text{ for } e
  \in E}$).

\begin{example}
\label{ex:triangle}
Given an even integer $N$, we construct an instance $I_N$
such that (1) $|R^{I_N}| = |S^{I_N}| = |T^{I_N}|=N$,
(2) $|R \Join S| = |R \Join T| = |S \Join T| = N^2/4 + N/2$, and 
(3) $|R \Join S \Join T| = 0$. The following instance satisfies all three
properties:
\[ R^{I_N} = S^{I_N} = T^{I_N} = 
\left\{(0,j)\right\}_{j=1}^{N/2} \cup 
\left\{(j,0)\right\}_{j=1}^{N/2}. 
\]
For example,
\[ R \Join S = \{(i,0,j)\}_{i,j=1}^{N/2} 
   \cup \{ (0,i,0) \}_{i=1,\dots,N/2} \]
and $R \Join S \Join T = \emptyset$. Thus, any standard join-based
algorithm takes time $\Omega(N^2)$.
We show later that AGM's algorithm takes $\Omega(N^2)$-time
too. Recall that the AGM bound
for this instance is $O(N^{3/2})$, and our algorithm thus takes time
$O(N^{3/2})$. In fact, as shall be shown later,
on this particular family of instances both of our
algorithms take only $O(N)$ time.
\end{example}

\section{Connections to Geometric Inequalities}
\label{sec:bt:equiv}

We describe the Bollob\'as-Thomason (BT) inequality from discrete
geometry and prove that BT inequality is equivalent to AGM's
inequality. We then look at a special case of BT inequality, the
Loomis-Whitney (LW) inequality, from which our algorithmic development
starts in the next section. We state the BT inequality:

\begin{thm}[Discrete Bollob\'as-Thomason (BT) Inequality]\label{thm:BT}
Let $S\subset \mathbb Z^n$ be a finite set of $n$-dimensional grid points.
Let $\mathcal F$ be a collection of subsets of $[n]$ in which every
$i \in [n]$ occurs in exactly $d$ members of $\mathcal F$.
Let $S_F$ be the set of projections $\mathbb Z^n \to \mathbb Z^F$ of
points in $S$ onto the coordinates in $F$. Then,
$|S|^d \leq \prod_{F \in \mathcal F} |S_F|$.
\end{thm}

To prove the equivalence between BT inequality and the AGM bound, we
first need a simple observation.

\begin{lmm}\label{lmm:whytight}\label{LMM:WHYTIGHT}
Consider an instance of the OJ problem consisting of a hypergraph $H=(V,E)$,
a fractional cover $\mv x=(x_e)_{e\in E}$ of $H$, and relations
$R_e$ for $e\in E$. Then, in linear time we can
transform the instance into another instance $H'=(V,E')$, 
$\mv x'=(x'_e)_{e\in E'}$, $(R'_e)_{e\in E'}$, such that
the following properties hold:
\bi
\item[(a)] $\mv x'$ is a ``tight'' fractional edge cover of the 
hypergraph $H'$, namely $\mv x'\geq 0$ and
\[ \sum_{e\in E': v\in e} x'_e = 1, \ \ \text{ for every } v\in V. \]
\item[(b)] The two problems have the same answer:
\[ \Join_{e\in E} R_e = \ \Join_{e\in E'} R'_e. \]
\item[(c)] AGM's bound on the transformed instance is at least
as good as that of the original instance:
\[ \prod_{e\in E'} |R'_e|^{x'_e} \leq \prod_{e\in E} |R_e|^{x_e}. \]
\ei
\end{lmm}
\bp
We describe the transformation in steps. At each step properties (b) and (c)
are kept as invariants. After all steps are done, (a) holds.

While there still exists some vertex $v\in V$ such that
$\sum_{e\in E: v \in e} x_e > 1$, i.e. $v$'s constraint is not tight,
let $f$ be an arbitrary hyperedge $f\in E$ such that $v\in f$.
Partition $f$ into two parts $f = f_t \cup f_{\neg t}$, where
$f_t$ consists of all vertices $u\in f$ such that $u$'s constraint
is tight, and $f_{\neg t}$ consist of vertices $u\in f$ such that
$u$'s constraint is not tight.  Note that $v\in f_{\neg t}$.

Define
$\rho = \min\left\{x_f, 
      \min_{u\in f_{\neg t}} \left\{ \sum_{e : u\in e} x_e - 1\right\}\right\}. 
$
This is the amount which, if we were able to reduce $x_f$ by $\rho$
then we will either turn $x_f$ to $0$ or make some constraint
for $u \in f_{\neg t}$ tight.
However, reducing $x_f$ might violate some already tight
constraint $u\in f_t$.
The trick is to ``break" $f$ into two parts.

We will set $E'=E\cup\{f_t\}$,
create a ``new'' relation $R'_{f_t} = \pi_{f_t}(R_f)$, and
keep all the old relations $R'_e=R_e$ for all $e\in E$.
Set the variables $x'_e=x_e$ for all $e\in E-\{f\}$ also.
The only two variables which have not been set are
$x'_f$ and $x'_{f_t}$. We set them as follows.

\bi
\item
When $x_f \leq  \min_{u\in f_{\neg t}} \left\{ \sum_{e : u\in e} x_e -
1\right\}$,
set
$x'_f = 0$ and
$x'_{f_t} = x_f$.
\item When $x_f > \min_{u\in f_{\neg t}} \left\{ \sum_{e : u\in e} x_e -
1\right\}$,
set
$x'_f = x_f - \rho$ and $x'_{f_t} = \rho$.
\ei

Either way, it can be readily verified that the new instance is a
legitimate OJ instance satisfying properties (b) and (c).
In the first case, some positive variable in some non-tight constraint has been
reduced to $0$.
In the second case, at least one non-tight constraint has become tight.
Once we change a variable $x_f$ (essentially ``break'' it up into
$x'_{f_t}$ and $x'_f$) we won't touch it again.
Hence, after a linear number of steps in $|V|$, we will have all tight
constraints.
\ep

With this
technical observation, we can now connect the two families of inequalities:

\begin{prop} BT inequality and AGM's fractional cover bound are
equivalent.
\end{prop}
\bp
To see that AGM's inequality implies BT inequality, we think of each 
coordinate as 
an attribute, and the projections $S_F$ as the input relations.
Set $x_F = 1/d$ for each $F\in \mathcal F$. It follows that
$\mv x = (x_F)_{F\in \mathcal F}$ is a fractional cover for
the hypergraph $H=([n], \mathcal F)$. AGM's bound then implies
that $|S| \leq \prod_{F\in\mathcal F} |S_F|^{1/d}$.

Conversely,
consider an instance of the OJ problem with hypergraph $H=(V,E)$
and a rational fractional cover $\mv x = (x_e)_{e\in E}$ of $H$.
First, by Lemma \ref{lmm:whytight},
we can assume that all cover constraints are tight, i.e.,
\[ \sum_{e: v\in e} x_e = 1, \ \ \text{ for any } v \in V. \]
By standard arguments, it can be shown that all the ``new" $x_e$ are rational values (even if the original values were not).
Second, by writing all variables $x_e$ as $d_e/d$ for a positive
common denominator $d$ we obtain
\[ \sum_{e: v\in e} d_e = d, \ \ \text{ for any } v \in V. \]
Now, create $d_e$ copies of each relation $R_e$.
Call the new relations $R'_e$. We obtain a new
hypergraph $H'=(V,E')$ where every attribute $v$ occurs in exactly $d$
hyperedges. This is precisely the Boll\'obas-Thomason's setting
of Theorem \ref{thm:BT}. Hence, the size of the join is bounded
above by
$ \prod_{e\in E'}|R'_e|^{1/d} = \prod_{e\in E}|R_e|^{d_e/d} = 
   \prod_{e\in E}|R_e|^{x_e}.
$
\ep

\paragraph*{Loomis-Whitney} We now consider a special case of BT (or
AGM), the discrete version of a classic geometric inequality called
the {\em Loomis-Whitney inequality} \cite{MR0031538}. The setting is
that for $n \geq 2$, $V = [n]$ and $E = \binom{V}{|V|-1}$. In this
case $x_e=1/(|V|-1), \forall e\in E$ is a fractional cover solution
for $(V,E)$, and LW showed the following:

\begin{thm}[Discrete Loomis-Whitney (LW) inequality]\label{thm:LW}
Let $S\subset \mathbb Z^n$ be a finite set of $n$-dimensional grid points.
For each dimension $i \in [n]$, let $S_{[n]\setminus \{i\}}$ denote the
$(n-1)$-dimensional projection of $S$ onto the coordinates
$[n]\setminus \{i\}$.  Then,
$|S|^{n-1} \leq \prod_{i=1}^n |S_{[n]\setminus \{i\}}|$.
\end{thm}

It is clear from our discussion above that LW is a special case of
BT (and so AGM), and it is with this special case that we begin
our algorithmic development in the next section.

%% file: main_results.tex

\newcommand{\lf}[1]{\mathcal{L}(#1)}
\algsetup{indent=2em}
\renewcommand{\algorithmiccomment}[1]{// {\em #1}}


\section{Algorithm for Loomis-Whitney instances}
\label{sec:LW-algo}

We first consider queries whose forms are slightly more general than
that in our motivating example (\ref{ex:triangle}). 
This class of queries has the same setup as in LW inequality
of Theorem \ref{thm:LW}.
In this spirit, we define a {\em Loomis-Whitney
(LW) instance} of the OJ problem to be a hypergraph $H=(V,E)$ such that
$E$ is the collection of all subsets of $V$ of size $|V| - 1$. When the
LW inequality is applied to this setting, it guarantees
that $|\Join_{e\in E} R_e| \leq \left(\prod_{e\in
  E}N_e\right)^{1/(n-1)}$, and the bound is tight in the
worst case. The main result of this section is the following:

\begin{thm}[Loomis-Whitney instance]
\label{thm:LW-n-1w}\label{THM:LW-N-1W}
Let $n\ge 2$ be an integer. 
Consider a Loomis-Whitney instance $H=(V=[n], E)$ of the OJ problem
with input relations $R_e$, where $|R_e|=N_e$ for $e\in E$.
Then the join $\Join_{e\in E} R_e$ can be computed in time
\[ O\left(n^2\cdot \left(\prod_{e\in E}N_e \right)^{1/(n-1)}+ 
          n^2\sum_{e\in E} N_e\right).
\]
\label{thm:lw:alg}
\end{thm}

Before proving this result, we give an example that illustrates the
intuition behind our algorithm and solve the motivating example from
the introduction (\ref{eq:q}).

\begin{example}\label{ex:lw:one}
  Recall that our input has three relations $R(A,B)$, $S(B,C)$,
  $T(A,C)$ and an instance $I$ such that $|R^I| = |S^I| = |T^I| = N$. 
  Let $J = R \Join S \Join T$. 
  Our goal is to construct $J$ in time $O(N^{3/2})$. For
  exposition, define a parameter $\tau \geq 0$ that we will choose
  below. We use $\tau$ to define two sets that effectively partition
  the tuples in $R^I$.
\[ D = \{  t_{B} \in \pi_{B}(R) : |R^{I}[t_{B}] | > \tau \} \text{ and } G = \{ (t_A,t_B) \in R^I : t_B \not\in D \} \]

\noindent
Intuitively, $D$ contains the heavy join keys in $R$. Note that $|D|< N/\tau$. Observe that
$J \subseteq (D \times T) \cup (G \Join S)$ (also note that this union
is disjoint). Our algorithm will construct $D \times T$ (resp. $G \Join
S$) in time $O(N^{3/2})$, then it will filter out those tuples in
both $S$ and $R$ (resp. $T$) using the hash tables on $S$ and $R$ 
(resp. $T$); this process produces exactly $J$. 
Since our running time is linear in the
above sets, the key question is how big are these two sets?

Observe that $|D \times T| \leq (N/ \tau) N = N^2/\tau$ while $|G
\Join S| = \sum_{t_B \in \pi_{B}(G) } |R[t_B]| |S[t_B]| \leq \tau
N$. Setting $\tau = \sqrt{N}$ makes both
terms at most $N^{3/2}$ establishing the running time of our
algorithm. One can check that if the relations are of different
cardinalities, then we can still use the same algorithm; moreover, by
setting $\tau = \sqrt{\frac{|R||T|}{|S|}}$, we achieve a running time
of $O(\sqrt{|R||S||T|} + |R| + |S| + |T|)$.
\end{example}

To describe the general algorithm underlying Theorem~\ref{thm:lw:alg}, we need to
introduce some data structures and notation.

\paragraph*{Data Structures and Notation} 
Let $H=(V,E)$ be an LW instance. Algorithm~\ref{algo:LW} begins by
constructing a labeled, binary tree $\T$ whose set of leaves is
exactly $V$ and each internal node has exactly two children. Any
binary tree over this leaf set can be used.  We denote the left child
of any internal node $x$ as $\lc(x)$ and its right child as
$\rc(x)$. Each node $x \in \T$ is labeled by a function $\lbl$, where
$\lbl(x) \subseteq V$ are defined inductively as follows: $\lbl(x) = V
\setminus \{x\}$ for a leaf node $x\in V$, and $\lbl(x) = \lbl(\lc(x))
\cap \lbl(\rc(x))$ if $x$ is an internal node of the tree.  It is
immediate that for any internal node $x$ we have $\lbl(\lc(x)) \cup
\lbl(\rc(x)) = V$ and that $\lbl(x) = \emptyset$ if and only if $x$ is
the root of the tree.  Let $J$ denote the output set of tuples of the
join, i.e. $J = \Join_{e \in E} R_e$.
For any node $x \in \T$, let $\T(x)$ denote the subtree of $\T$ rooted
at $x$, and $\lf{\T(x)}$ denote the set of leaves under this subtree.
For any three relations $R$, $S$, and $T$,
define  $R \Join_S T = (R \Join T) \lJoin S$.

\paragraph*{Algorithm for LW instances} 
Algorithm~\ref{algo:LW} works in two stages. 
Let $u$ be the root of the tree $\T$.
First we compute a tuple set $\deter{u}$ containing the output 
$J$ such that $\deter{u}$ has a relatively small size 
(at most the size bound times $n$).
Second, we prune those tuples that cannot
participate in the join (which takes only linear time in the size of
$\deter{u}$). 
The interesting part is how we compute $\deter{u}$. 
Inductively, we compute a set $\deter{x}$ that at each stage contains 
candidate tuples and an auxiliary set $\ndeter{x}$, 
which is a superset of the projection
$\pi_{\lbl(x)}(J \setminus \deter{x})$. 
The set $\ndeter{x}$ will intuitively allow 
us to deal with those tuples that would blow up the size of an intermediate
relation. The key novelty in Algorithm~\ref{algo:LW} is the
construction of the set $G$ that contains all those tuples (join keys)
that are in some sense {\em light}, i.e., joining over them
would not exceed the size/time bound $P$ by much. 
The elements that are not light are postponed to be processed later
by pushing them to the set $\ndeter{x}$.  This is in full analogy to the
sets $G$ and $D$ defined in Example~\ref{ex:lw:one}.

\begin{algorithm}[t]
\begin{algorithmic}[1]
\STATE An LW instance: $R_{e}$ for $e \in {V \choose |V| -1}$ and
$N_{e} = |R_e|$.  
\STATE $P = \prod_{e \in E} N_e^{1/(n-1)}$ (the size
bound from LW inequality)
\STATE $u \la \text{root}(\T)$; $(\deter{u}, \ndeter{u}) \la {\sf LW}(u)$
\STATE ``Prune'' $\deter{u}$ and return
\end{algorithmic}

\textsf{LW}$(x)$ : $x \in \T$ returns $(C,D)$ 
\begin{algorithmic}[1]
\IF{$x$ is a leaf}
  \RETURN $(\emptyset,R_{\lbl(x)})$
\ENDIF
\STATE $(C_L, D_L) \la {\sc LW}(\lc(x))$ and $(C_R, D_R) \la {\sc LW}(\rc(x))$
\STATE $F \la \pi_{\lbl(x)}(D_L) \cap \pi_{\lbl(x)}(D_R)$ 
\STATE $G \la \left\{ \mv t \in F : |D_{L}[\mv t]| + 1 \leq \lceil P / |D_R| \rceil \right\}$
\COMMENT{$F =G= \emptyset$ if $|D_R|=0$}
\IF{$x$ is the root of $\T$}
\STATE $C \la (D_L \Join D_R) \cup C_L \cup C_R$ 
\STATE $D \la \emptyset$
\ELSE
\STATE $C \la (D_L \Join_G D_R) \cup C_L \cup C_R$ 
\STATE $D \la F \setminus G$.
\ENDIF
\RETURN $(C,D)$
\end{algorithmic}
\caption{Algorithm for Loomis-Whitney Instances}
\label{algo:LW}
\end{algorithm}

\bp[Proof of Theorem~\protect{\ref{thm:LW-n-1w}}]
We claim that the following three properties hold for every node $x \in \T$: 
\bi
\item[(1)] $\pi_{\lbl(x)}(J \setminus \deter{x}) \subseteq \ndeter{x}$;
\item[(2)] $|\deter{x}|\le (|\lf{\T(x)}|-1) \cdot P$; and
\item[(3)] $|\ndeter{x}|\le \min\left\{ \min_{l \in \lf{\T(x)}} \{ N_{[n]\setminus\{l \}}\}, \frac{\prod_{l  \in \lf{\T(x)}} N_{[n]\setminus\{l \}}}{P^{|\lf{\T(x)}|-1}} \right\}$.
\ei
Assuming these three properties hold, let us first prove that that 
Algorithm~\ref{algo:LW} correctly computes the join, $J$. Let $u$ denote the root of the tree
$\T$. By property (1),
\begin{eqnarray*}
\pi_{\lbl(\lc(u))}(J \setminus \deter{\lc(u)}) &\subseteq& \ndeter{\lc(u)}\\
\pi_{\lbl(\rc(u)))}(J \setminus \deter{\rc(u)}) &\subseteq& \ndeter{\rc(u)}
\end{eqnarray*}
Hence, 
\[ J \setminus (\deter{\lc(u)} \cup \deter{\rc(u)}) \subseteq 
   \ndeter{\lc(u)} \times \ndeter{\rc(u)}  = 
   \ndeter{\lc(u)} \Join \ndeter{\rc(u)}.
\]
This implies $J \subseteq \deter{u}$. Thus, from 
$\deter{u}$ we can compute $J$ by keeping only tuples in $\deter{u}$ 
whose projection on any attribute set $e\in E = \binom{[n]}{n-1}$ is 
contained in $R_e$ (the ``pruning'' step).

\newcommand{\rchld}[1]{\rc(#1)}
\newcommand{\lchld}[1]{\lc(#1)}

We next show that properties 1-3 hold by induction on each step of
Algorithm~\ref{algo:LW}. For the base case, consider
$\ell\in\lf{\T}$. Recall that in this case $\deter{\ell}=\emptyset$
and $\ndeter{\ell} = R_{[n]- \{\ell\}}$; thus, properties 1-3 hold.

Now assume that properties 1-3 hold for all children of an internal node $v$.
We first verify properties 2-3 for $v$. From the definition of $G$, 
\[ |\ndeter{\rchld{v}} \Join_G \ndeter{\lchld{v}}|\le 
   \left(\left\lceil \frac{P}{|\ndeter{\rchld{v}}|}\right\rceil-1\right)
   \cdot |\ndeter{\rchld{v}}|\le P. 
\] 
From the inductive upper bounds on $\deter{\lchld{v}}$ and $\deter{\rchld{v}}$, 
property 2 holds at $v$. 
By definition of $G$ and a straightforward counting argument, note that 
\[ |\ndeter{v}|=|F\setminus G| \le |\ndeter{\lchld{v}}|\cdot 
   \frac{1}{\lceil P/|\ndeter{\rchld{v}}|\rceil}\le 
   \frac{|\ndeter{\lchld{v}}| \cdot |\ndeter{\rchld{v}}|}{P}. \]
From the induction hypotheses on $\lchld{v}$ and $\rchld{v}$, we have
\begin{eqnarray*}
|\ndeter{\lchld{v}}| &\le& \frac{\prod_{\ell\in\lf{\T(\lchld{v})}} N_{[n]-\{\ell\}}}{P^{|\lf{\T(\lchld{v})}|-1}}\\
|\ndeter{\rchld{v}}| &\le& \frac{\prod_{\ell\in\lf{\T(\rchld{v})}} N_{[n]-\{\ell\}}}{P^{|\lf{\T(\rchld{v})}|-1}},
\end{eqnarray*}
which implies that 
\[ |\ndeter{{v}}|\le \frac{\prod_{\ell\in\lf{\T({v})}} N_{[n]-\{\ell\}}}{P^{|\lf{\T({v})}|-1}}.
\] 
Further, it is easy to see that 
\[ |\ndeter{{v}}|\le \min(|\ndeter{\lchld{v}}|, |\ndeter{\rchld{v}}|), \] 
which by induction implies that 
\[ |\ndeter{v}|\le \min_{\ell\in \lf{\T(v)}} N_{[n]-\{\ell\}}. \]
Property 3 is thus verified.

Finally, we verify property 1. By induction, we have 
\begin{eqnarray*}
\pi_{\lbl(\lchld{v})}(J\setminus\deter{\lchld{v}}) &\subseteq& 
  \ndeter{\lchld{v}}\\
\pi_{\lbl(\rchld{v})}(J\setminus \deter{\rchld{v}}) &\subseteq& 
\ndeter{\rchld{v}}
\end{eqnarray*}
This along with the fact that $\lbl(\lchld{v}) \cap \lbl(\rchld{v}) = \lbl(v)$ 
implies that 
\[ \pi_{\lbl(v)}(J\setminus \deter{\lchld{v}} \cup \deter{\rchld{v}}) 
   \subseteq \ndeter{\lchld{v}}_{\lbl(v)}\cap \ndeter{\rchld{v}}_{\lbl(v)} =
    G\uplus D(v). 
\]
Further, every tuple in 
$(J\setminus\deter{\lchld{v}}\cup\deter{\rchld{v}})$ whose projection onto 
$\lbl(v)$ is in $G$ also belongs to 
$\ndeter{\rchld{v}} \Join_G \ndeter{\lchld{v}}$. 
This implies that 
$\pi_{\lbl(v)}(J\setminus \deter{v})= \ndeter{v}$, 
as desired.

For the run time complexity of Algorithm~\ref{algo:LW}, we claim that for
every node $x$, we need time $O(n|\deter{x}|+n|\ndeter{x}|)$.  To see
this note that for each node $x$, the lines 4, 5, 7, and 9 of
the algorithm can be computed in that much time using hashing. 
Using property (3) above, we have a (loose) upper bound of 
$O\left(nP+ n\min_{l  \in \lf{\T(x)}} N_{[n]\setminus\{l \}}\right)$ on
the run time for node $x$. 
Summing the run time over all the nodes in
the tree gives the claimed run time.
\ep

\section{An Algorithm for All Join Queries}
\label{sec:all:j}\label{SEC:ALL:J}


This section presents our algorithm for proving the AGM's inequality
with running time matching the bound.

\bthm
Let $H=(V,E)$ be a hypergraph representing a natural join query.
Let $n=|V|$ and $m=|E|$.
Let $\mv x = (x_e)_{e\in E}$ be an arbitrary point in the fractional 
cover polytope
\begin{eqnarray*}
\sum_{e : v \in e} x_e &\geq& 1, \ \text{ for any $v \in V$}\\
x_e &\geq &0, \text{ for any $e \in E$}
\end{eqnarray*}
For each $e\in E$, let $R_e$ be a relation of size $N_e=|R_e|$
(number of tuples in the relation). Then,
\bi
 \item[(a)] The join $\Join_{e\in E}R_e$ has size (number of tuples) bounded by
\[ |\Join_{e\in E}R_e| \leq \prod_{e\in E}N_e^{x_e}. \]
 \item[(b)] Furthermore, the join $\Join_{e\in E}R_e$ can be computed
in time
\[ O\left(mn\prod_{e\in E}N_e^{x_e} + n^2\sum_{e\in E}N_e + m^2n\right) \]
\ei
\label{thm:main}
\ethm

\brmk In the running time above, $m^2n$ is the query preprocessing time,
$n^2\sum_{e\in E}N_e$ is the data preprocessing time, and 
$mn\prod_{e\in E}N_e^{x_e}$ is the query evaluation time.
If all relations in the database are indexed in advance
to satisfy three conditions (HTw), $w\in\{1,2,3\}$, below,
then we can remove the term $n^2\sum_{e\in E}N_e$ from the running time.
Also, the fractional cover solution $\mv x$ should probably be the
best fractional cover in terms of the linear objective
$\sum_e (\log N_e)\cdot x_e$.
The data-preprocessing time of $O(n^2\sum_eN_e)$ is for a single known query.
If we were to index all relations in advance without knowing which
queries to be evaluated, then the advance-indexing takes
$O(n\cdot n! \sum_e N_e)$-time.
This price is paid once, up-front, for an arbitrary number of
future queries.
\ermk

Before turning to our algorithm and proof of this theorem, we observe
that a consequence of this theorem is the following algorithmic
version of the discrete version of BT inequality.

\bcor
Let $S\subset \mathbb Z^n$ be a finite set of $n$-dimensional grid points.
Let $\mathcal F$ be a collection of subsets of $[n]$ in which every
$i \in [n]$ occurs in exactly $d$ members of $\mathcal F$.
Let $S_F$ be the set of projections $\mathbb Z^n \to \mathbb Z^F$ of
points in $S$ onto the coordinates in $F$. Then,
\begin{equation}
 |S|^d \leq \prod_{F \in \mathcal F} |S_F|. 
\label{eqn:BT}
\end{equation}
Furthermore, given the projections $S_F$ we can compute $S$ in time
\[ O\left( |\mathcal F|n\left(\prod_{F \in \mathcal F} |S_F|\right)^{1/d}
           + n^2\sum_{F\in \mathcal F}|S_F|+|\mathcal F|^2n \right) \]
\ecor

Recall that the LW inequality is a special case of the BT inequality.
Hence, our algorithm proves the LW inequality as well.

\subsection{Main ingredients of the algorithm}

There are three key ingredients in the algorithm
(Algorithm \ref{algo:the-algo}) and its analysis:
\be
 \item We first build a ``search tree'' for each relation $R_e$ which will be 
used throughout the algorithm. 
We can also build a collection of hash indices which functionally can 
serve the same purpose. We use the ``search tree'' data structure here to
make the analysis clearer.
This step is responsible for the (near-) linear 
term $O(n^2\sum_{e\in E}N_e)$ in the running time. 
The search tree for each relation is built using
a particular ordering of attributes in the relation called the total order.
The total order is constructed from a data structure called 
a query plan tree which also drives the recursion structure of the algorithm.
 \item Suppose we have two relations $A$ and $B$ on the same set of attributes
and we'd like to compute $A\cap B$. If the search trees for $A$ and $B$
have already been built, the intersection can be computed in time
$O(k\min\{|A|,|B|\})$ where $k$ is the number of attributes in $A$,
because we can traverse every tuple of
the smaller relation and check into the search structure for the larger
relation. Also note that, for any two non-negative numbers 
$a$ and $b$ such that $a+b\geq 1$, we have
$\min\{|A|, |B|\} \leq |A|^a|B|^b$. 
 \item The third ingredient is based on 'unrolling' sums using 
generalized H\"older inequality \eqref{ineq:Holder}
in a correct way. We cannot explain it in  a few lines and thus will 
resort to an example presented in the next section. 
The example should give the reader the correct intuition into the 
entire algorithm and its analysis without getting lost
in heavy notations.
\ee

We make extensive use of the following form of H\"older's inequality
which was also attributed to Jensen.
(See the classic book ``Inequalities'' by
Hardy, Littlewood, and P\'olya \cite{MR89d:26016},
Theorem 22 on page 29.)

\begin{lmm}[Hardy, Littlewood, and P\'olya~\cite{MR89d:26016}]  
\label{lem:holder}
Let $m,n$ be positive integers. Let $y_1,\dots,y_n$ be non-negative
real numbers such that $y_1+\cdots+y_n\geq 1$.  Let $a_{ij} \geq 0$ be
non-negative real numbers, for $i\in [m]$ and $j\in [n]$.  With the
convention $0^0 = 0$, we have:
\begin{equation}
\sum_{i=1}^m\prod_{j=1}^n a_{ij}^{y_j}
\leq
\prod_{j=1}^n\left(\sum_{i=1}^m a_{ij}\right)^{y_j}.
\label{ineq:Holder}
\end{equation}
\end{lmm}

For each tuple $\mv t$ on attribute set $A$, we will
write $\mv t$ as $\mv t_A$ to emphasize the support of $\mv t$:
$\mv t_A = (t_a)_{a\in A}$.
Consider any
relation $R$ with attribute set $S$.  Let $A\subset S$ and $\mv t_A$
be a fixed tuple.
Then, $\pi_A(R)$ denote the projection of $R$ down to attributes in $A$.
And, define the {\em $\mv t_A$-section} of $R$ to be
\[ R[\mv t_A] := \pi_{S-A}(R \lJoin \{\mv t_A\})
   = \{ \mv t_{S-A} \suchthat (\mv t_A, \mv t_{S-A}) \in R\}. 
\]
In particular, $R[\mv t_\emptyset] = R$.

\subsection{A complete worked example for our algorithm and its analysis}
\label{subsec:an-example}

Before presenting the algorithm and analyze it formally, we first work out a
small query to explain how the algorithm and the analysis works.
It should be noted that the following example does not cover all
the intricacies of the general algorithm, especially in the boundary cases.
We aim to convey the intuition first.
Also, the way we label nodes in the QP-tree in this example is slightly
different from the way nodes are labeled in the general algorithm,
in order to avoid heavy sub-scripting.

Consider the following instance to the OJ problem.
The hypergraph $H$ has 6 attributes $V = \{1,\dots,6\}$, and 5 relations
$R_a,R_b,R_c,R_d,R_e$ defined by the following vertex-edge incident matrix
$\mv M$:.
\[
\mv M =
\begin{array}{l||ccccc}
 & a & b & c & d & e\\
\hline
\hline
\rowcolor{blue}
1& 1 & 1 & 1 & 0 & 0\\
\rowcolor{lightgray}
2& 1 & 0 & 1 & 1 & 0\\
\rowcolor{green}
3& 0 & 1 & 1 & 0 & 1\\
\rowcolor{lightgray}
4& 1 & 1 & 0 & 1 & 0\\
\rowcolor{green}
5& 1 & 0 & 0 & 0 & 1\\
\rowcolor{green}
6& 0 & 1 & 0 & 1 & 1
\end{array}
\]

We are given a fractional cover solution 
$\mv x=(x_a,x_b,x_c,x_d,x_e)$, i.e.
$\mv M\mv x \geq \mv 1$.

{\bf Step 0.} 
We first build something called a {\em query plan tree} (QP-tree).
The tree has nodes labeled by the hyperedge $a,b,c,d,e$,
except for the leaf nodes each of which can be labeled by a subset of 
hyperedges.
(Note again that the labeling in this example is slightly different from the 
labeling done in the general algorithm's description to avoid
cumbersome notations.)
Each node of the query plan tree also has an associated {\em universe} 
which is a subset of attributes.
The reader is referred to Figure \ref{fig:example-join-tree} for 
an illustration of the tree building process.
In the figure, the universe for each node is drawn next to the parent edge
to the node.

\begin{figure}[t]
\centerline{\includegraphics[width=4in]{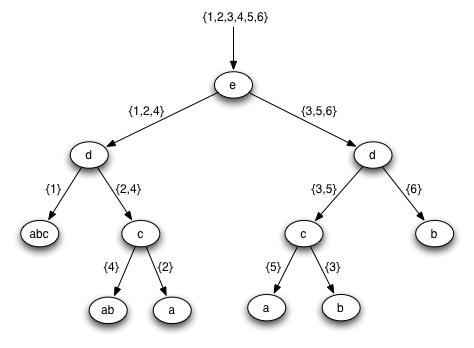}}
\label{fig:example-join-tree}
\caption{A query plan tree for the example OJ instance}
\end{figure}

The query plan tree is built recursively as follows.
We first arbitrarily order the hyperedges. In the example shown in Figure
\ref{fig:example-join-tree},
we have built a tree with the order $e,d,c,b,a$.
The root node has universe $V$. We visit these edges one by one
in that order.

If every remaining hyperedges contains 
the universe $V$ then label the node with all remaining hyperedges and stop.
In this case the node is a leaf node.
Otherwise, consider the next hyperedge in the visiting order
above (it is $e$ as we are in the beginning).
Label the root with $e$, and
create two children of the root $e$.
The left child will have universe $V-e$, and the right child has
$e$ as its universe. 
Now, we recursively build the left tree starting from the next 
hyperedge (i.e. $d$) in the ordering, 
but only restricting to the smaller universe
$\{1,2,4\}$. 
Similarly, we build the right tree starting from the next hyperedge ($d$) in the
ordering, but only restricting to the smaller universe $\{3,5,6\}$.

Let us explain one more level of the tree building process to make things clear.
Consider the left tree of the root node $e$.
The universe is $\{1,2,4\}$. The root node will be the next hyperedge $d$
in the ordering. But we really work on the restriction of
$d$ in the universe $\{1,2,4\}$, which is $d'=d\cap\{1,2,4\}=\{2,4\}$.
Then, we create two children. The left child has universe
$\{1,2,4\}-d' = \{1\}$. The right child has universe $d'=\{2,4\}$.
For the left child, the universe has size $1$ and all three remaining
hyperedges $a$, $b$, and $c$ contain $1$, hence we label the left
child with $abc$.

By visiting all leaf nodes from left to right and print 
the attributes in their universes,
we obtain something called {\em the total order} of all attributes in $V$. 
In the figure, the total order is $1,4,2,5,3,6$.
(In the general case, the total order is slightly more complicated than
in this example. See Procedure \ref{algo:total-order}.)

Finally, based on the total order $1,4,2,5,3,6$ just obtained,
we build search trees for all relations respecting this ordering.
For relation $R_a$, the top level of the tree is indexed over attribute
$1$, the next two levels are $4$ and $2$, and the last level
is indexed over attribute $5$.
For $R_b$, the order is $1, 4, 3, 6$.
For $R_c$, the order is $1, 2, 3$.
For $R_d$, the order is $4, 2, 6$.
For $R_e$, the order is $5, 3, 6$.
It will be clear later that the attribute orders in the search trees 
have a decisive effect on the overall running time.
This is also the step that is responsible for the 
term $O(n^2 \sum_e N_e)$ in the overall running time.

{\bf Step 1.} 
(This step corresponds to the left most node of the query plan tree.) 
Compute the join 
\begin{equation}
 T_1 = \pi_{\{1\}}(R_a) \Join \pi_{\{1\}}(R_b) \Join \pi_{\{1\}}(R_c) 
\label{eqn:T1-join}
\end{equation}
as follows. 
This is the join over attributes {\em not} in $d$ and $e$.
If $|\pi_{\{1\}}(R_a)|$ is the smallest among 
$|\pi_{\{1\}}(R_a)|$, $|\pi_{\{1\}}(R_b)|$, and
$|\pi_{\{1\}}(R_c)|$, then for each attribute 
$t_1 \in \pi_{\{1\}}(R_a)$, we search the first levels of the search trees
for $R_b$ and $R_c$ to see if $t_1$ is in both
$\pi_{\{1\}}(R_b)$ and $\pi_{\{1\}}(R_b)$.
Similarly, if $|\pi_{\{1\}}(R_b)|$ or 
$|\pi_{\{1\}}(R_c)|$ is the smallest then 
for each $t_1 \in \pi_{\{1\}}(R_b)$ (or in $\pi_{\{1\}}(R_c)$) we search 
for attribute $t_1$ in the other two search trees.
As attribute $1$ is in the first level of all three search trees,
the join \eqref{eqn:T1-join} can be computed  in time
\[ O(|T_1|) = O\left(\min\left\{|\pi_{\{1\}}(R_a)|,
|\pi_{\{1\}}(R_b)|,
|\pi_{\{1\}}(R_c)|\right\}\right).
\]
Note that 
\[ |T_1| \leq \min\left\{|\pi_{\{1\}}(R_a)|,
|\pi_{\{1\}}(R_b)|,
|\pi_{\{1\}}(R_c)|\right\}
\leq 
|\pi_{\{1\}}(R_a)|^{x_a}
|\pi_{\{1\}}(R_b)|^{x_b}
|\pi_{\{1\}}(R_c)|^{x_c}
\leq
N_a^{x_a}N_b^{x_b}N_c^{x_c}
\]
because $x_a+x_b+x_c\geq 1$.
In particular, step 1 can be performed within the run-time budget.

{\bf Step 2.}
(This step corresponds to the node labeled $d$ on the left branch
of query plan tree.)
Compute the join 
\[ T_{\{1,2,4\}} = 
\pi_{\{1, 2,4\}}(R_a) \Join
\pi_{\{1,4\}}(R_b) \Join
\pi_{\{1,2\}}(R_c) \Join
\pi_{\{2,4\}}(R_d)
\]
This is a join over all attributes {\em not} in $e$.

Since we have already computed the join $T_1$ over attribute $1$
of $R_a$, $R_b$, and $R_c$, the relation
$T_{\{1,2,4\}}$ can be computed by computing
for every $t_1\in T_1$ the $t_1$-section of $T_{\{1,2,4\}}$
\[ T_{\{1,2,4\}}[t_1] = 
\underbrace{\pi_{\{2,4\}}(R_a[t_1])}_{A[t_1]} \Join
\underbrace{\pi_{\{4\}}(R_b[t_1])}_{B[t_1]} \Join
\underbrace{\pi_{\{2\}}(R_c[t_1])}_{C[t_1]} \Join
\underbrace{\pi_{\{2,4\}}(R_d)}_{D}
\]
and then $T_{\{1,2,4\}}$ is simply the union of all the $t_1$-sections
$T_{\{1,2,4\}}[t_1]$.
The notations $A[t_1]$, $B[t_1]$, $C[t_1]$, and $D$ are defined for
the sake of brevity.

Fix $t_1 \in T_1$, we next describe how $T_{\{1,2,4\}}[t_1]$ is computed.
If $x_d\geq 1$ then we go directly to case 2b below.
When $x_d<1$, define
\begin{eqnarray*}
x'_a&=&\frac{x_a}{1-x_d}\\
x'_b&=&\frac{x_b}{1-x_d}\\
x'_c&=&\frac{x_c}{1-x_d}.
\end{eqnarray*}
Consider the hypergraph graph $H'$ which is the graph $H$ restricted
to the vertices $2,4$ and edges $a,b,c$.
In particular, $H'$ has vertex set $\{2,4\}$ and edges $\{2,4\}, \{4\}, \{2\}$.
It is clear that $x'_a$,$x'_b$, and $x'_c$ form a fractional cover solution
of $H'$ because $\mv x$ was a fractional cover solution for $H$.
Thus, $H'$, $\mv x'=(x'_a,x'_b,x'_c)$, and $A[t_1]$, $B[t_1]$,
and $C[t_1]$ form an instance of the OJ problem.
We will recursively solve this instance if a condition is satisfied.

{\bf Case 2a.}
Suppose
\[ |A[t_1]|^{x'_a} |B[t_1]|^{x'_b} |C[t_1]|^{x'_c} \leq |D| \]
then we (recursively) compute the join
$A[t_1]\Join B[t_1]\Join C[t_1]$. 
By induction on the instance $H'$, this  join can be computed in time
\[ O\left(|A[t_1]|^{x'_a} |B[t_1]|^{x'_b} |C[t_1]|^{x'_c} \right). \]
(This induction hypothesis corresponds to the node labeled $c$
on the left branch of the query plan tree.)
Here, we crucially use the fact that the search trees for $R_a$,
$R_b$, $R_c$ have been built so that the subtrees under the branch
$t_1$ are precisely the search trees for relations $A[t_1], B[t_1], C[t_1]$
and thus are readily available to compute this join.
Now, to get $T_{\{1,2,4\}}[t_1]$ we simply check whether every 
tuple in $A[t_1]\Join B[t_1]\Join C[t_1]$ belongs to $D$.

{\bf Case 2b.}
Suppose either $x_d\geq 1$ or
\[ |D| \leq |A[t_1]|^{x'_a} |B[t_1]|^{x'_b} |C[t_1]|^{x'_c} \]
then for every tuple $(t_2,t_4)$ in $D$ we check whether 
$(t_2,t_4)\in A[t_1]$, $t_4\in B[t_1]$, {\em and}
$t_2\in C[t_1]$. The overall running time is $O(|D|)$.

Thus, for a fixed value $t_1$, the relation $T_{\{1,2,4\}}[t_1]$
can be computed in time
\[ O\left(
   \min\{|A[t_1]|^{x'_a} |B[t_1]|^{x'_b} |C[t_1]|^{x'_c}, |D|\}
\right).
\]
In fact, it is not hard to see that
\[ |T_{\{1,2,4\}}[t_1]| \leq 
\min\{|A[t_1]|^{x'_a} |B[t_1]|^{x'_b} |C[t_1]|^{x'_c}, |D|\}. \]
This observation will eventually imply the inequality \eqref{eqn:agm08-bound} 
(for this instance), and in the general case leads to the constructive
proof of the inequality \eqref{eqn:agm08-bound}.

Next, note that
\begin{eqnarray*}
\min\left\{|A[t_1]|^{x'_a} |B[t_1]|^{x'_b} |C[t_1]|^{x'_c}, |D| \right\}
&\leq&
\left(|A[t_1]|^{x'_a} |B[t_1]|^{x'_b} |C[t_1]|^{x'_c}\right)^{1-x_d} |D|^{x_d}\\
&=& 
|A[t_1]|^{x_a} |B[t_1]|^{x_b} |C[t_1]|^{x_c} |D|^{x_d}.
\end{eqnarray*}
If $x_d\geq 1$ then the run-time is also in the order of
$O(|A[t_1]|^{x_a} |B[t_1]|^{x_b} |C[t_1]|^{x_c} |D|^{x_d})$.
Consequently, the total running time for step 2 is in the order of
\begin{eqnarray*}
\sum_{t_1\in T_1} |A[t_1]|^{x_a} |B[t_1]|^{x_b} |C[t_1]|^{x_c} |D|^{x_d}
&=&
|D|^{x_d} \sum_{t_1\in T_1} |A[t_1]|^{x_a} |B[t_1]|^{x_b} |C[t_1]|^{x_c}\\
&\leq&
|D|^{x_d} 
\left(\sum_{t_1\in T_1} |A[t_1]|\right)^{x_a} 
\left(\sum_{t_1\in T_1} |B[t_1]|\right)^{x_b} 
\left(\sum_{t_1\in T_1} |C[t_1]|\right)^{x_c} \\
&\leq&
|D|^{x_d} 
\cdot |\pi_{\{1,2,4\}}(R_a)|^{x_a}
\cdot |\pi_{\{1,4\}}(R_b)|^{x_b}
\cdot |\pi_{\{1,2\}}(R_c)|^{x_c}\\
&\leq& N_a^{x_a}N_b^{x_b}N_c^{x_c}N_d^{x_d}.
\end{eqnarray*}
The first inequality follows from generalized H\"older inequality because
$x_a+x_b+x_c\geq =1$ and $x_a,x_b,x_c \geq 0$.
The second inequality says that if we sum over the sizes of the $t_1$-sections,
we get at most the size of the relation.
In summary, step $2$ is still within the running time budget.

{\bf Step 3.} Compute the final join over all attributes
\[ T_{\{1,2,3,4,5,6\}} = 
R_a \Join R_b \Join R_c \Join R_d \Join R_e. \]
Since we have already computed the join $T_{\{1,2,4\}}$ over attributes
$1,2,4$
of $R_a$, $R_b$, $R_c$, and $R_d$, the relation
$T_{\{1,2,3,4,5,6\}}$ can be computed by computing
for every $(t_1,t_2,t_4) \in T_{\{1,2,4\}}$ the join
\[
 T_{\{1,2,3,4,5,6\}}[t_1,t_2,t_4] =
\underbrace{\pi_{\{5\}}(R_a[t_1,t_2,t_4])}_{A} \Join
\underbrace{\pi_{\{3,6\}}(R_b[t_1,t_4])}_{B} \Join
\underbrace{\pi_{\{3\}}(R_c[t_1,t_2])}_{C} \Join
\underbrace{\pi_{\{6\}}(R_d[t_2,t_4])}_{D} \Join
\underbrace{R_e}_{E},
\]
and return the union of these joins over all 
tuples $(t_1,t_2,t_4) \in T_{\{1,2,4\}}$.
Again, the notations $A$, $B$, $C$, $D$, $E$ are introduced to
for the sake of brevity. Note, however, that they are different
from the $A$, $B$, $C$, $D$ from case 2.
This step illustrates the third ingredient of the algorithm's analysis.

Fix $(t_1,t_2,t_4) \in T_{\{1,2,4\}}$.
If $x_e\geq 1$ then we jump to case 3b; otherwise, define
\begin{eqnarray*}
x''_a&=& \frac{x_a}{1-x_e}\\
x''_b&=& \frac{x_b}{1-x_e}\\
x''_c&=& \frac{x_c}{1-x_e}\\
x''_d&=& \frac{x_d}{1-x_e}.
\end{eqnarray*}
Then define a hypergraph $H''$ on the attributes
$\{3,5,6\}$ and the restrictions of $a$, $b$, $c$, $d$
on these attributes. Clearly the vector $\mv x''$ is a fractional cover
for this instance. 

{\bf Case 3a}.
Suppose $x_e\geq 1$ or
\[ |A|^{x''_a} |B|^{x''_b} |C|^{x''_c} |D|^{x''_d} \leq |E|. \]
By applying the induction hypothesis on the $H''$ instance we can
compute the join $A \Join B \Join C \Join D$ in time
$O\left(|A|^{x''_a} |B|^{x''_b} |C|^{x''_c} |D|^{x''_d}\right)$.
(The induction hypothesis corresponds to the node labeled $d$ on
{\em right} branch of the query plan tree.)
Again, because the search trees for all relations have been built
in such a way that the search trees for $A$, $B$, $C$, $D$
are already present on $t_1,t_2,t_4$-branches of the trees
for $R_a,R_b,R_c$, and $R_d$, there is no extra time spent on indexing
for computing this join.
Then, for every tuple $\mv t_{\{3,5,6\}}$ in the join we check
(the search tree for) $E$ to see if the tuple belongs to $E$.

{\bf Case 3b}.
Suppose
\[ |E| \leq |A|^{x''_a} |B|^{x''_b} |C|^{x''_c} |D|^{x''_d}. \]
Then, for each tuple $\mv t_{\{3,5,6\}} =(t_3,t_5,t_6) \in E$ we check
to see whether $t_5 \in A, (t_3,t_6)\in B, t_3 \in C$, {\em and}
$t_6 \in D$.

Either way, for a fix tuple $(t_1,t_2,t_4) \in T_{\{1,2,4\}}$
the running time is
\[ \tilde O\left(
\min\left\{
|A|^{x''_a} |B|^{x''_b} |C|^{x''_c} |D|^{x''_d},
|E|\right\}
\right). \]
Now, we apply the same trick as in case 2:
\begin{eqnarray*}
\min\left\{
|A|^{x''_a} |B|^{x''_b} |C|^{x''_c} |D|^{x''_d},
|E|\right\}
&\leq&
\left(
|A|^{x''_a} |B|^{x''_b} |C|^{x''_c} |D|^{x''_d}
\right)^{1-x_e}|E|^{x_e}\\
&=&
|A|^{x_a}|B|^{x_b}|C|^{x_c}|D|^{x_d}|E|^{x_e}\\
&\leq&
|R_a[t_1,t_2,t_4]|^{x_a}
|R_b[t_1,t_4]|^{x_b}
|R_c[t_1,t_2]|^{x_c}
|R_d[t_2,t_4]|^{x_d}
|R_e|^{x_e}.
\end{eqnarray*}
Hence, the total running time for step $3$ is in the order of
\begin{eqnarray*}
&&\sum_{(t_1,t_2,t_4)\in T_{\{1,2,4\}}}
|R_a[t_1,t_2,t_4]|^{x_a}
|R_b[t_1,t_4]|^{x_b}
|R_c[t_1,t_2]|^{x_c}
|R_d[t_2,t_4]|^{x_d}
|R_e|^{x_e}\\
&=&
|R_e|^{x_e}
\sum_{t_1}
\sum_{t_2}
\sum_{t_4}
|R_a[t_1,t_2,t_4]|^{x_a}
|R_b[t_1,t_4]|^{x_b}
|R_c[t_1,t_2]|^{x_c}
|R_d[t_2,t_4]|^{x_d}
\end{eqnarray*}
where the first sum is over $t_1\in \pi_{\{1\}}(T_{\{1,2,4\}})$,
the second sum is over $t_2$ such that
$(t_1,t_2)\in \pi_{\{1,2\}}(T_{\{1,2,4\}})$,
and the third sum is over $t_4$
such that $(t_1,t_2,t_4)\in T_{\{1,2,4\}}$.
We apply H\"older inequality several times to ``unroll" the sums.
Note that we crucially use the fact that $\mv x$ is a fractional
cover solution ($\mv M\mv x \geq \mv 1$) to apply H\"older's inequality.
\begin{eqnarray*}
&&|R_e|^{x_e} \sum_{t_1} \sum_{t_2} \sum_{t_4}
|R_a[t_1,t_2,t_4]|^{x_a}
|R_b[t_1,t_4]|^{x_b}
|R_c[t_1,t_2]|^{x_c}
|R_d[t_2,t_4]|^{x_d}\\
&=&|R_e|^{x_e} \sum_{t_1} \sum_{t_2} 
|R_c[t_1,t_2]|^{x_c}
\sum_{t_4}
|R_a[t_1,t_2,t_4]|^{x_a}
|R_b[t_1,t_4]|^{x_b}
|R_d[t_2,t_4]|^{x_d}\\
&\leq&|R_e|^{x_e} \sum_{t_1} \sum_{t_2} 
|R_c[t_1,t_2]|^{x_c}
\left(\sum_{t_4}|R_a[t_1,t_2,t_4]|\right)^{x_a}
\left(\sum_{t_4}|R_b[t_1,t_4]|\right)^{x_b}
\left(\sum_{t_4}|R_d[t_2,t_4]|\right)^{x_d}\\
&\leq&|R_e|^{x_e} \sum_{t_1} \sum_{t_2} 
|R_c[t_1,t_2]|^{x_c}
|R_a[t_1,t_2]|^{x_a}
|R_b[t_1]|^{x_b}
|R_d[t_2]|^{x_d}\\
&=&|R_e|^{x_e} \sum_{t_1} 
|R_b[t_1]|^{x_b}
\sum_{t_2} 
|R_c[t_1,t_2]|^{x_c}
|R_a[t_1,t_2]|^{x_a}
|R_d[t_2]|^{x_d}\\
&\leq&|R_e|^{x_e} \sum_{t_1} 
|R_b[t_1]|^{x_b}
\left(\sum_{t_2} |R_c[t_1,t_2]|\right)^{x_c}
\left(\sum_{t_2} |R_a[t_1,t_2]|\right)^{x_a}
\left(\sum_{t_2} |R_d[t_2]|\right)^{x_d}\\
&\leq&|R_e|^{x_e} \sum_{t_1} 
|R_b[t_1]|^{x_b}
|R_c[t_1]|^{x_c}
|R_a[t_1]|^{x_a}
|R_d|^{x_d}\\
&=&|R_e|^{x_e} |R_d|^{x_d}
\sum_{t_1} 
|R_b[t_1]|^{x_b}
|R_c[t_1]|^{x_c}
|R_a[t_1]|^{x_a}\\
&\leq&|R_e|^{x_e} |R_d|^{x_d}
\left(\sum_{t_1} |R_b[t_1]|\right)^{x_b}
\left(\sum_{t_1} |R_c[t_1]|\right)^{x_c}
\left(\sum_{t_1} |R_a[t_1]|\right)^{x_a}\\
&\leq&|R_e|^{x_e} |R_d|^{x_d} |R_b|^{x_b} |R_c|^{x_c} |R_a|^{x_a}.
\end{eqnarray*}

\subsection{Rigorous description and analysis of the algorithm}
\label{app:analysis-of-the-algo}

Algorithm \ref{algo:the-algo} computes the join of $m$ given relations.
Beside the relations, 
the input to the algorithm consists of the hypergraph $H=(V,E)$
with $|V|=n$, $|E|=m$, and a point $\mv x=(x_e)_{e\in E}$ in the 
fractional cover polytope
\begin{eqnarray*}
\sum_{v \in e} x_e &\geq& 1, \ \text{for any $v \in V$}\\
x_e &\geq &0, \text{for any $e \in E$}.
\end{eqnarray*}

\be
\item We first build a {\em query plan tree}. The query plan tree serves
two purposes: (a) it captures the structure of the recursions 
in the algorithm where each node of the tree roughly corresponds to a 
sub-problem,
(b) it gives a total order of all the attributes based on which we
can pre-build search trees for all the relations in the next step.
\item From the query plan tree, we construct a total order of all
attributes in $V$. Then, for each relation $R_e$ we construct
a search tree for $R_e$ based on the relative order of $R_e$'s attributes
imposed by the total order.
\item We traverse the query plan tree and solve some of the sub-problems
and combine the solutions to form the final answer.
It is important to note that {\bf not all} sub-problems corresponding to
nodes in the query plan trees will be solved. We decide whether to solve
a sub-problem based on a ``size check." Intuitively, if the sub-problem
is estimated to have a large output size we will try to not solve it.
\ee

\renewcommand{\algorithmicrequire}{\textbf{Input:}}
\renewcommand{\algorithmicensure}{\textbf{Output:}}
\begin{algorithm}[h]
\caption{Computing the join $\Join_{e\in E}R_e$}
\begin{algorithmic}[1]
\REQUIRE Hypergraph $H=(V,E)$, $|V|=n$, $|E|=m$
\REQUIRE Fractional cover solution $\mv x = (x_e)_{e\in E}$
\REQUIRE Relations $R_e, e\in E$
\STATE Compute the query plan tree $\T$, let $u$ be $\T$'s root node
\STATE Compute a total order of attributes 
\STATE Compute a collection of hash indices for all relations
\RETURN {\sc Recursive-Join}$(u, \mv x, \nil)$
\end{algorithmic}
\label{algo:the-algo}
\end{algorithm}

We repeat some of the terminologies already defined so that
this section is relatively self-contained.
For each tuple $\mv t$ on attribute set $A$, we will write
$\mv t$ as $\mv t_A$ to signify the fact that the tuple is on the attribute
set $A$: $\mv t_A = (t_a)_{a\in A}$.
Consider any relation $R$ with attribute set $S$.
Let $A\subset S$ and $\mv t_A$ be a fixed tuple.
Then $R[\mv t_A]$ denotes the ``$\mv t_A$-section" of $R$,
which is a relation on $S-A$
consisting of {\em all} tuples $\mv t_{S-A}$ such that
$(\mv t_A, \mv t_{S-A}) \in R$.
In particular, $R[\mv t_\emptyset] = R$.
Let $\pi_A(R)$ denote the projection
of $R$ down to attributes in $A$.

\subsubsection{Step (1): Building the query plan tree}
\label{subsubsec:build-QP-tree}

\begin{algorithm}
\begin{algorithmic}[1]
\STATE Fix an arbitrary order $e_1, e_2, \dots, e_m$ of all the
       hyperedges in $E$.
\STATE $\T \la $ {\sc build-tree}$(V, m)$
\end{algorithmic}

{\sc build-tree}$(U, k)$
\begin{algorithmic}[1]
\IF {$e_i \cap U = \emptyset, \forall i\in[k]$}
 \RETURN $\nil$
\ENDIF
\STATE Create a node $u$ with $\lbl(u) \la k$ and $\univ(u) = U$
\IF {$k>1$ and $\exists i \in [k]$ such that $U\not\subseteq e_i$}
   \STATE $\lc(u) \la $ {\sc build-tree}$(U\setminus e_k, k-1)$
   \STATE $\rc(u) \la $ {\sc build-tree}$(U\cap e_k, k-1)$
\ENDIF
\RETURN $u$
\end{algorithmic}
\caption{Constructing the query plan tree $\T$}
\label{algo:QPT}
\end{algorithm}

Very roughly, each node $x$ and the sub-tree below it forms the
``skeleton'' of a sub-problem. There will be many sub-problems that 
correspond to each skeleton. 
The value $\lbl(x)$ points to an ``anchor'' relation for the sub-problem
and $\univ(x)$ is the set of attributes that the sub-problem is joining
on. The anchor relation divides the universe $\univ(x)$ into two
parts to further sub-divide the recursion structure.
Fix an arbitrary order $e_1, e_2, \dots, e_m$ of 
{\em all} the hyperedges in $E$.
For notational convenience, for any $k\in [m]$ define
$E_k = \{e_1,\dots,e_k\}$.
The query plan tree $\T$ is a binary tree with the following associated
information:
\bi
 \item {\em Labels}. Each node of $\T$ has a ``label" $\lbl(u)$  which is an
integer $k\in [m]$.
 \item {\em Universes}. Each node $u$ of $\T$ has a ``universe"
$\univ(u)$ which is a non-empty subset of attributes: $\univ(u)\subseteq V$.
 \item Each internal node $u$ of $\T$ has a left child $\lc(u)$
or a right child $\rc(u)$ or both. If a child does not exist then the
child pointer points to $\nil$.
\ei
Algorithm \ref{algo:QPT} builds the query plan tree $\T$.
Very roughly, each node $x$ and the sub-tree below it forms the
``skeleton'' of a sub-problem. There will be many sub-problems that 
correspond to each skeleton. 
The value $\lbl(x)$ points to an ``anchor'' relation for the sub-problem
and $\univ(x)$ is the set of attributes that the sub-problem is joining
on. The anchor relation divides the universe $\univ(x)$ into two
parts to further sub-divide the recursion structure.

Note that line 5 and 6 will not be executed if
$U\subseteq e_i, \forall i\in [k]$,
in which case $u$ is a leaf node.
When $u$ is not a leaf node, if $U\subseteq e_k$ then $u$ will not
have a left child ($\lc(u)=\nil$).
The running time for this pre-processing step 
is $O(m^2n)$.

Figure \ref{fig:sample-QPT} shows a query plan tree
produced by Algorithm \ref{algo:QPT} on an example query.

\begin{figure}
\begin{center}
\begin{tabular}{cc}
\parbox[b]{1.75in}{
\begin{align*}
q & = & R_1(A_1,A_2, A_4, A_5)\\
& \Join & R_2(A_1,A_3,A_4,A_6)\\ 
& \Join &R_3(A_1,A_2,A_3)\\ 
& \Join & R_4(A_2,A_4,A_6)\\
& \Join & R_5(A_3,A_5,A_6)
\end{align*}}&
\includegraphics[width=2.0in]{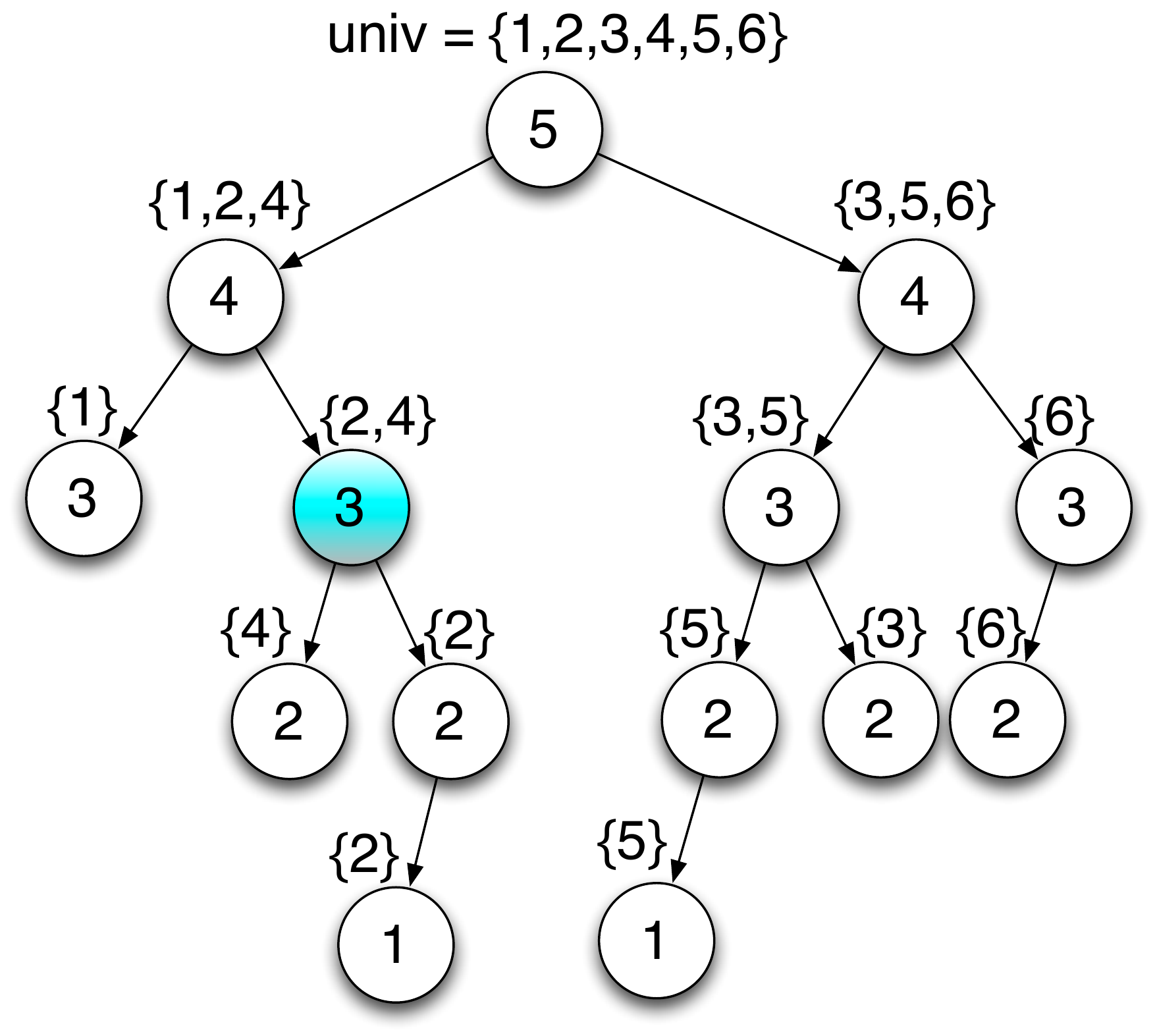} 
\end{tabular}
\end{center}
\caption{(a) A query $q$ and (b) a sample QP tree for $q$.}
\label{fig:sample-QPT}
\end{figure}

\subsubsection{Step (2): Computing a total order of the attributes and 
building the search trees}
\label{subsubsec:subsec:total-order}

From the query plan tree $\T$, Procedure \ref{algo:total-order}
constructs a total order of all the attributes in $V$.
We will call this ordering {\em the total order} of $V$.
It is not hard to see that the total order satisfies the following 
proposition.
\bprop
The total order computed in Algorithm 
\ref{algo:total-order} satisfies the following
properties
\bi
 \item[(TO1)] For every node $u$ in the query plan tree $\T$, all members
              of $\univ(u)$ are consecutive in the total order
 \item[(TO2)] For every internal node $u$, if $\lbl(u) = k$ and
$S$ is the set of all attributes preceding $\univ(u)$ in the total
order, then $S\cup\univ(\lc(u)) = S\cup(U\setminus e_k)$ is precisely
the set of all attributes preceding $\univ(\rc(u)) = e_k\cap U$ in the 
total order.
\ei
\label{prop:TO-properties}
\eprop

\begin{algorithm}
\caption{Computing a total order of attributes in $V$}
\label{algo:total-order}
\begin{algorithmic}[1]
\STATE Let $\T$ be the query plan tree with root node $u$, where $\univ(u)=V$
\STATE {\sc print-attribs}$(u)$
\end{algorithmic}
{\sc print-attribs}$(u)$
\begin{algorithmic}[1]
\IF {$u$ is a leaf node of $\T$}
 \PRINT all attributes in $\univ(u)$ in an arbitrary order
\ELSIF {$\lc(u) = \nil$}
 \STATE {\sc print-attribs}$(\rc(u))$
\ELSIF {$\rc(u) = \nil$}
 \STATE {\sc print-attribs}$(\lc(u))$
 \PRINT all attributes in $\univ(u) \setminus \univ(\lc(u))$ in an arbitrary 
        order
\ELSE
 \STATE {\sc print-attribs}$(\lc(u))$
 \STATE {\sc print-attribs}$(\rc(u))$
\ENDIF
\end{algorithmic}
\end{algorithm}

For each relation $R_e$, $e\in E$, we order all attributes in $R_e$
such that the internal order of attributes in $R_e$ is consistent
with the total order of all attributes computed by Algorithm 
\ref{algo:total-order}.
More concretely, suppose $R_e$ has $k$ attributes ordered
$a_1,\dots,a_k$, then $a_i$ must come before $a_{i+1}$ in the total order,
for all $1\leq i\leq k-1$.
Then, we build a search tree (or any indexing data structure) for
every relation $R_e$ using the internal order of $R_e$'s attributes:
$a_1$ indexes level $1$ of the tree, $a_2$ indexes the next level, $\dots$,
$a_k$ indexes the last level of the tree.
The search tree for relation $R_e$
is constructed to satisfy the following three properties.
Let $i$ and $j$ be arbitrary integers such that $1\leq i\leq j\leq k$.
Let $\mv t_{\{a_1,\dots,a_i\}} = (t_{a_1},\dots,t_{a_i})$ 
be an arbitrary tuple on the attributes $\{a_1,\dots,a_i\}$.
\bi
 \item[(ST1)] We can decide whether 
$\mv t_{\{a_1,\dots,a_i\}} \in \pi_{\{a_1,\dots,a_i\}}(R_e)$
in $O(i)$-time (by ``stepping down" the tree along the $t_{a_1},\dots,t_{a_i}$
path).
 \item[(ST2)] We can query the size
$|\pi_{\{a_{i+1},\dots,a_j\}}(R_e[\mv t_{\{a_1,\dots,a_i\}}])|$
in $O(i)$ time.
 \item[(ST3)] We can list all tuples in the set
$\pi_{\{a_{i+1},\dots,a_j\}}(R_e[\mv t_{\{a_1,\dots,a_i\}}])$
in time linear in the output size if the output is not empty.
\ei

The total running time for building all the search trees is
$O(n^2\sum_e N_e)$.

\subsubsection{Step (3): Computing the join}
\label{subsubsec:subsec:computing-the-join}

\floatname{algorithm}{Procedure}

\begin{algorithm}[t]
\caption{{\sc Recursive-Join}$(u, \mv y, \mv t_S)$}
\label{proc:rec-join}
\begin{algorithmic}[1]
\STATE Let $U = \univ(u)$, $k = \lbl(u)$
\STATE $\ret \la \emptyset$ \COMMENT{$\ret$ is the returned tuple set}
\IF[note that $U\subseteq e_i, \forall i \leq k$]{$u$ is a leaf node of $\T$}
  \STATE $j \la \argmin_{i\in [k]} \left\{ |\pi_U(R_{e_i}[\mv t_{S\cap e_i}])| \right\}$
  \STATE \COMMENT{By convention, $R_e[\nil] = R_e$ and $R_e[\mv t_\emptyset] = R_e$}
  \FOR {each tuple $\mv t_U \in \pi_U(R_{e_j}[\mv t_{S\cap e_j}])$}
    \IF{$\mv t_U \in \pi_U(R_{e_i}[\mv t_{S\cap e_i}]), \text{ for all } i \in [k]\setminus\{j\}$}
      \STATE $\ret \la \ret \cup \{(\mv t_S, \mv t_U)\}$
    \ENDIF
  \ENDFOR
  \RETURN $\ret$
\ENDIF
\IF[$u$ is not a leaf node of $\T$]{$\lc(u) = \nil$}
    \STATE $L \la \{\mv t_S\}$ 
    \STATE \COMMENT{note that $L \neq \emptyset$ and
$\mv t_S$ could be $\nil$ (when $S=\emptyset$)}
  \ELSE
    \STATE $L \la $ {\sc Recursive-Join}$(\lc(u), (y_1,\dots,y_{k-1}), \mv t_S)$
  \ENDIF
  \STATE $W \la U \setminus e_k$, $W^- \la e_k \cap U$
  \IF{$W^- = \emptyset$}
    \RETURN $L$
  \ENDIF
    \FOR{each tuple $\mv t_{S\cup W} =(\mv t_S, \mv t_W) \in L$}
  \IF{$y_{e_k} \geq 1$}
    \STATE {\bf go to} line 27
  \ENDIF
    \IF {$\displaystyle{\left(\prod_{i=1}^{k-1}
      |\pi_{e_i \cap W^-}(R_{e_i}[\mv t_{(S\cup W)\cap e_i}])|^{\frac{y_{e_i}}{1-y_{e_k}}} < |\pi_{W^-}(R_{e_k}[\mv t_{S\cap e_k}])|
                    \right)}$}
         \STATE $Z \la $ {\sc Recursive-Join}$\left(\rc(u), \left(\frac{y_{e_i}}{1-y_{e_k}}\right)_{i=1}^{k-1}, \mv t_{S\cup W}\right)$
         \FOR {each tuple $(\mv t_S, \mv t_W, \mv t_{W^-}) \in Z$}
           \IF{$\mv t_{W^-} \in \pi_{W^-}(R_{e_k}[\mv t_{S\cap e_k}])$}
              \STATE $\ret \la \ret \cup \{(\mv t_S, \mv t_W, 
                      \mv t_{W^-})\}$
           \ENDIF
         \ENDFOR
       \ELSE
         \FOR {each tuple $\mv t_{W^-} \in \pi_{W^-}(R_{e_k}[\mv t_{S\cap e_k}])$}
           \IF{$\mv t_{e_i\cap W^-} \in \pi_{e_i \cap W^-}(R_{e_i}[\mv t_{(S\cup W) \cap e_i}])$ for all  $e_i$ such that $i<k$ and $e_i \cap W^- \neq \emptyset$}
              \STATE $\ret \la \ret \cup \{(\mv t_S, \mv t_W, 
                      \mv t_{W^-})\}$
           \ENDIF
         \ENDFOR
    \ENDIF
  \ENDFOR
\RETURN $\ret$
\end{algorithmic}
\end{algorithm}

At the heart of Algorithm \ref{algo:the-algo} is a recursive
procedure called {\sc Recursive-Join} (Procedure \ref{proc:rec-join})
which takes three arguments: 
\bi
 \item a node $u$ from the query plan tree $\T$ whose label is $k$ 
       for some $k\in [m]$.
 \item a fractional cover solution $\mv y_{E_k} = (y_{e_1},\dots,y_{e_k})$
       of the hypergraph $(\univ(u), E_k)$.
       Here, we only take the restrictions of 
       hyperedges of $E_k$ onto the universe $\univ(u)$. Specifically,

\begin{eqnarray*}
\sum_{e \in E_k: i \in e} y_e &\geq& 1, \ \text{for any $i \in \univ(u)$}\\
y_e &\geq &0, \text{for any $e \in E_k$}
\end{eqnarray*}

 \item a tuple $\mv t_{S} = (t_i)_{i\in S}$ where  $S$ is the set of {\em all}
       attributes in $V$ which precede $\univ(u)$ in the total order.
       (Due to property (TO1) of Proposition \ref{prop:TO-properties}, 
        the set $S$ is well-defined.)
       If there is no attribute preceding $\univ(u)$ then this argument
       is $\nil$.
       In particular, the argument is $\nil$ if $u$ is a node along the left
       path of QP-tree $\T$ from the root down to the left-most leaf.
\ei

Throughout this section, we denote the final output by $J$ which
is defined to be $J = \Join_{e \in E} R_e$.
The goal of {\sc Recursive-Join} is to compute a superset of the
relation $\{\mv t_S\} \times \pi_{\univ(u)}(J[\mv t_S])$, 
i.e., a superset of the
output tuples that start with $\mv t_S$ on the attributes $S \cup
\univ(u)$. This intermediate output is analogous to the set $C$ in
Algorithm~\ref{algo:LW} for LW instances.  A second similarity
Algorithm~\ref{algo:LW} is that our algorithm makes a choice per tuple
based on the output's estimated size.

Theorem \ref{thm:main} is a special case of the following lemma
where we set $u$ to be the root of the QP-tree $\T$,
$\mv y = \mv x$, and $S=\emptyset$ ($\mv t_S=\nil$).
Finally, we observe that we need only 
$O(n^2)$ number of hash indices per input relation, 
which completes the proof.

\blmm \label{lmm:main}
Consider a call {\sc Recursive-Join}$(u, \mv y, \mv t_S)$ to
Procedure \ref{proc:rec-join}.
Let $k=\lbl(u)$ and $U=\univ(u)$. Then,
\bi
 \item[(a)] The procedure outputs a relation $\ret$ on attributes $S\cup U$
with at most the following number of tuples
\[ B(u,\mv y, \mv t_S) := \prod_{i=1}^k |\pi_{U \cap e_i}(R_{e_i}[\mv t_{S\cap e_i}])|^{y_i}. \]
(For the sake of presentation, we agree on the convention that
when $U\cap e_i=\emptyset$ we set
$|\pi_{U \cap e_i}(R_{e_i}[\mv t_{S\cap e_i}])|=1$ so that the factor
does not contribute anything to the product.)
 \item[(b)] Furthermore, the procedure runs in time
$O(mn \cdot B(u,\mv y,\mv t_S))$.
\ei
\label{app:lmm:main}
\elmm
\bp
We prove both $(a)$ and $(b)$ by induction on the height of the subtree
of $\T$ rooted at $u$. The proof will also explain in ``plain'' English
the algorithm presented in Procedure \ref{proc:rec-join}.
The procedure tries to compute the join
\[ \{\mv t_S\} \times 
   \left( \Join_{i=1}^k \pi_{U\cap e_i}(R_{e_i}[\mv t_{S\cap e_i}]) \right).
\]
Roughly speaking, it is computing the join of all the sections
$R_{e_i}[\mv t_{S\cap e_i}]$ inside the universe $U$.

{\bf Base case}. The height of the sub-tree rooted at $u$ is zero,
i.e. $u$ is a leaf node. In this case, lines 4-9 of Procedure
\ref{proc:rec-join} is executed.
When $u$ is a leaf node, $U \subseteq e_i, \forall i \in [k]$
and thus $U = U\cap e_i, \forall i\in [k]$.
Since $\mv y$ is a fractional cover solution to the 
hypergraph instance $(U, E_k)$, we know $\sum_{i=1}^ky_i \geq 1$.
The join has size at most
\[ \min_{i\in [k]} \left\{ |\pi_U(R_{e_i}[\mv t_{S\cap e_i}])| \right\}
   \leq \prod_{i=1}^k |\pi_{U\cap e_i}(R_{e_i}[\mv t_{S\cap e_i}])|^{y_i}
   = B(u, \mv y, \mv t_S).
\]
To compute the join, we go over each tuple of the smallest-sized 
section-projection
$\pi_U(R_{e_j}[\mv t_{S\cap e_j}])$ and check to see if the tuple
belongs to all the other section-projections.
There are at most $k$ other sections, and due to property
(ST1) each check takes time $O(n)$.
Hence, the total time spent is
$O(mn\cdot B(u,\mv y, \mv t_S))$.

{\bf Inductive step}. Now, consider the case when $u$ is not a leaf node.

If $\lc(u)=\nil$ which means $U\subseteq e_k$ then there is no
attribute in $U\setminus e_k$ to join over (line 11).
Otherwise, we first recursively call the ``left sub-problem''
(Line 14) and store the result in $L$.
Note that the attribute set of $L$ is $S\cup W = S\cup (U\setminus e_k)$.
We need to verify that the arguments we gave to this recursive call
are legitimate. It should be obvious that $k-1=\lbl(\lc(u))$.
Since $\mv y=(y_1,\dots,y_k)$ is a fractional cover of
the $(U,E_k)$ hypergraph, $\mv y'=(y_1,\dots,y_{k-1})$ is a 
fractional cover of the $(U\setminus e_k, E_{k-1})$ hypergraph.
And, $\univ(\lc(u)) = U\setminus e_k$.
Finally, due to property (TO2) $S$ is precisely the set of attributes
preceding $\univ(\lc(u))$ in the total order.
From the induction hypothesis, the recursive call on line 14 takes time
\[ 
O(mn\cdot B(\lc(u), \mv y', \mv t_S)) = O\left(mn \prod_{i=1}^{k-1} 
  |\pi_{W\cap e_i}(R_{e_i}[\mv t_{S\cap e_i}])|^{y_i} \right).
\]
Furthermore, the number of tuples in $L$ is at most
$B(\lc(u), \mv y', \mv t_S) =
\prod_{i=1}^{k-1}   |\pi_{W\cap e_i}(R_{e_i}[\mv t_{S\cap e_i}])|^{y_i}$.

If $W^- = \emptyset$ then $L$ is returned and we are done
because in this case
$B(\lc(u),\mv y', \mv t_S) \leq B(u, \mv y, \mv t_S)$.

Consider the for loop from line 18 to line 29.
We execute the for loop for each tuple 
$\mv t_{S\cup W} = (\mv t_S, \mv t_W) \in L$.
If $L = \emptyset$ then the output is empty and we are done.
If $L = \{\mv t_S\}$ then this for-loop is executed only once.
This is the case if the assignment in line 11 was performed,
which means $U\subseteq e_k$ and thus $W=\emptyset$.
We do not have to analyze this case separately as it is subsumed
by the general case that $L\neq \emptyset$.

Note that if $y_{e_k}\geq 1$ then we go directly to {\bf case b} below
(corresponding to line 27).

{\bf Case a.} Consider the case when $y_{e_k}<1$ and
\[ \prod_{i=1}^{k-1}
      |\pi_{e_i \cap W^-}(R_{e_i}[\mv t_{(S\cup W)\cap e_i}])|^{\frac{y_{
e_i}}{1-y_{e_k}}} < |\pi_{W^-}(R_{e_k}[\mv t_{S\cap e_k}])|.
\]
In this case we first recursively solve the sub-problem
\[
Z = \text{\sc Recursive-Join}\left(\rc(u), \left(\frac{y_{e_i}}{1-y_{e_k}}\right)_{i=1}^{k-1}, \mv t_{S\cup W}\right).
\]
We need to make sure that the arguments are legitimate.
Note that $\univ(\rc(u)) = W^-$, and that $y_{e_k}<1$.
The sub-problem is on the hypergraph $(W^-, E_{k-1})$.
For any $v\in W^-=U\cap e_k$, because $\mv y$ is a fractional cover
of the $(U, E_k)$ hypergraph,
\[ 1 \leq \sum_{i \in [k] \ :  \ v \in e_i} y_{e_i} =
      y_{e_k} + \sum_{i\in [k-1] \ : \ v\in e_i} y_{e_i}.
\]
Hence,
\[ 1 \leq \sum_{i\in [k-1] \ : \ v\in e_i} \frac{y_{e_i}}{1-y_{e_k}},\]
which confirms that the solution 
$\left(\frac{y_{e_i}}{1-y_{e_k}}\right)_{i=1}^{k-1}$ is a fractional cover
for the hypergraph $(W^-, E_{k-1})$.
Finally, by property (TO2) the attributes $S\cup W$ are precisely the
attributes preceding $W^-$ in the total order.

After solving the sub-problem we obtain a tuple set $Z$ over the
attributes $S\cup W\cup W^- = S\cup U$. By the induction hypothesis
the time it takes to solve the sub-problem is
\[
O\left(mn\prod_{i=1}^{k-1} |\pi_{e_i \cap W^-}(R_{e_i}[\mv t_{(S\cup W)\cap e_i}])|^{\frac{y_{e_i}}{1-y_{e_k}}}
\right)
\]
and the number of tuples in $Z$ is bounded by
\[ \prod_{i=1}^{k-1} |\pi_{e_i \cap W^-}(R_{e_i}[\mv t_{(S\cup W)\cap e_i}])|^{\frac{y_{e_i}}{1-y_{e_k}}}.
\]
Then, for each tuple in $Z$ we perform the check on line 24. 
Hence, the overall running time in this case is still
$O\left(mn\prod_{i=1}^{k-1} |\pi_{e_i \cap W^-}(R_{e_i}[\mv t_{(S\cup W)\cap e_i}])|^{\frac{y_{e_i}}{1-y_{e_k}}}
\right)$

{\bf Case b.} Consider the case when either $y_{e_k}\geq 1$ or
\[ \prod_{i=1}^{k-1}
      |\pi_{e_i \cap W^-}(R_{e_i}[\mv t_{(S\cup W)\cap e_i}])|^{\frac{y_{
e_i}}{1-y_{e_k}}} \geq |\pi_{W^-}(R_{e_k}[\mv t_{S\cap e_k}])|.
\]
In this case, we execute lines 27 to 29. The number of tuples output
is at most $|\pi_{W^-}(R_{e_k}[\mv t_{S\cap e_k}])|$
and the running time is 
$O(mn|\pi_{W^-}(R_{e_k}[\mv t_{S\cap e_k}])|)$.

Overall, for both (case a) and (case b) the number of tuples output
is bounded by
\[ T = 
  \begin{cases}
    \min\left\{\prod_{i=1}^{k-1}
      |\pi_{e_i \cap W^-}(R_{e_i}[\mv t_{(S\cup W)\cap e_i}])|^{\frac{y_{e_i}}{1-y_{e_k}}}, |\pi_{W^-}(R_{e_k}[\mv t_{S\cap e_k}])| \right\} 
& y_{e_k}<1\\
|\pi_{W^-}(R_{e_k}[\mv t_{S\cap e_k}])| & \text{otherwise}
  \end{cases}
\]
and the running time is in the order of $O(mnT)$. We bound $T$ next.
When $y_{e_k}<1$ we have
\begin{eqnarray*}
T&\leq&\min\left\{\prod_{i=1}^{k-1}
    |\pi_{e_i \cap W^-}(R_{e_i}[\mv t_{(S\cup W)\cap e_i}])|^{\frac{y_{e_i}}{1-y_{e_k}}}, |\pi_{W^-}(R_{e_k}[\mv t_{S\cap e_k}])| \right\}\\
&\leq&\left(\prod_{i=1}^{k-1}
   |\pi_{e_i \cap W^-}(R_{e_i}[\mv t_{(S\cup W)\cap e_i}])|^{\frac{y_{e_i}}{1-y_{e_k}}} \right)^{1-y_{e_k}} |\pi_{W^-}(R_{e_k}[\mv t_{S\cap e_k}])|^{y_{e_k}}\\
&=&|\pi_{U\cap e_k}(R_{e_k}[\mv t_{S\cap e_k}])|^{y_{e_k}} \cdot \prod_{i=1}^{k-1}
|\pi_{e_i \cap W^-}(R_{e_i}[\mv t_{(S\cup W)\cap e_i}])|^{y_{e_i}}
\end{eqnarray*}
When $y_{e_k}\geq 1$, it is obvious that the same inequality holds:
\[ T \leq 
|\pi_{U\cap e_k}(R_{e_k}[\mv t_{S\cap e_k}])|^{y_{e_k}} \cdot \prod_{i=1}^{k-1}
|\pi_{e_i \cap W^-}(R_{e_i}[\mv t_{(S\cup W)\cap e_i}])|^{y_{e_i}}.
\]
Summing overall $(\mv t_S, \mv t_W)\in L$, the number of output tuples
is bounded by the following sum.
Without loss of generality, assume $W=\{1,\dots,d\}=[d]$.
In the following, the first sum is over
$t_1\in \pi_{\{1\}}(L)$, the second sum is over $t_2$ such that
$(t_1,t_2)\in \pi_{\{1,2\}}(L)$, and so on.
To shorten notations a little, define
\[ \bar R_i = R_{e_i}[\mv t_{S\cap e_i}]. \]
Then, the total number of output tuples is bounded by
\begin{eqnarray*}
&& \sum_{\mv t_W \in \pi_W(L)} |\pi_{U\cap e_k}(R_{e_k}[\mv t_{S\cap e_k}])|^{y_{e_k}} \cdot \prod_{i=1}^{k-1}
|\pi_{e_i \cap W^-}(R_{e_i}[\mv t_{(S\cup W)\cap e_i}])|^{y_{e_i}}\\
&=& |\pi_{U\cap e_k}(\bar R_k)|^{y_{e_k}} 
\sum_{t_1}\sum_{t_2}\cdots\sum_{t_d}
\prod_{i=1}^{k-1}
|\pi_{e_i \cap W^-}(\bar R_i[\mv t_{[d]\cap e_i}])|^{y_{e_i}}\\
&=& |\pi_{U\cap e_k}(\bar R_k)|^{y_{e_k}} 
\sum_{t_1}\cdots \sum_{t_{d-1}} 
\prod_{i<k, d\notin e_i}
|\pi_{e_i \cap W^-}(\bar R_i[\mv t_{[d]\cap e_i}])|^{y_{e_i}}
\sum_{t_d}
\prod_{i<k, d\in e_i}
|\pi_{e_i \cap W^-}(\bar R_i[\mv t_{[d]\cap e_i}])|^{y_{e_i}}\\
&\leq& |\pi_{U\cap e_k}(\bar R_k)|^{y_{e_k}} 
\sum_{t_1}\cdots \sum_{t_{d-1}} 
\prod_{i<k, d\notin e_i}
|\pi_{e_i \cap W^-}(\bar R_i[\mv t_{[d]\cap e_i}])|^{y_{e_i}}
\prod_{i<k, d\in e_i}
\left(\sum_{t_d}
|\pi_{e_i \cap W^-}(\bar R_i[\mv t_{[d]\cap e_i}])|\right)^{y_{e_i}}\\
&\leq& |\pi_{U\cap e_k}(\bar R_k)|^{y_{e_k}} 
\sum_{t_1}\cdots \sum_{t_{d-1}} 
\prod_{i<k, d\notin e_i}
|\pi_{e_i \cap (W^- \cup \{d\})}(\bar R_i[\mv t_{[d-1]\cap e_i}])|^{y_{e_i}}
\prod_{i<k, d\in e_i}
|\pi_{e_i \cap (W^- \cup \{d\})}(\bar R_i[\mv t_{[d-1]\cap e_i}])|^{y_{e_i}}\\
&=& |\pi_{U\cap e_k}(\bar R_k)|^{y_{e_k}} 
\sum_{t_1}\sum_{t_2}\cdots\sum_{t_{d-1}}
\prod_{i=1}^{k-1}
|\pi_{e_i \cap (W^- \cup\{d\})}(\bar R_i[\mv t_{[d-1]\cap e_i}])|^{y_{e_i}}\\
&\leq&\dots\\
&\leq& |\pi_{U\cap e_k}(\bar R_k)|^{y_{e_k}} 
\sum_{t_1}\sum_{t_2}\cdots\sum_{t_{d-2}}
\prod_{i=1}^{k-1}
|\pi_{e_i \cap (W^- \cup\{d-1,d\})}(\bar R_i[\mv t_{[d-2]\cap e_i}])|^{y_{e_i}}\\
&\leq&\dots\\
&=& \prod_{i=1}^{k} |\pi_{U\cap e_i}(\bar R_i)|^{y_{e_i}}
\end{eqnarray*}
\ep

\section{Limits of Standard Approaches}
\label{sec:limits}

For a given join query $q$, we describe a sufficient syntactic
condition for $q$ so that when computed by any join-project plan is
asymptotically slower than the worst-case bound. Our algorithm runs
within this bound, and so for such $q$ there is an asymptotic
running-time gap.

\paragraph*{LW Instances}

Recall that an {\em LW instance} of the OJ problem is a join query $q$
represented
by the hypergraph $(V,E)$, where $V=[n]$, and $E = \binom{[n]}{n-1}$ for some
integer $n \geq 2$.
Our main result in this section is the following lemma\footnote{We thank an
anonymous PODS'12 referee for giving us the argument showing that our example
works for all join-project plans rather than just the AGM algorithm and a
join-tree algorithm.}

\blmm
Let $n\geq 2$ be an arbitrary integer.
Given any LW-query $q$ represented by a hypergraph $([n], \binom{[n]}{n-1})$,
and any positive integer $N\geq 2$, there exist relations $R_i$, $i\in[n]$,
such that $|R_i| =N,\forall i\in [n]$, the attribute set for $R_i$ is
$[n]-\{i\}$, and that {\em any} join-project plan for $q$ on these relations
runs in time $\Omega(N^2/n^2)$.
\label{LEM:BAD:INSTANCE}
\elmm

Before proving the lemma, we note that both the traditional join-tree
algorithm and AGM's algorithm are join-project plans, and thus their
running times are asymptotically worse than the best AGM bound for this
instance which is
$|\Join_{i=1}^n R_i|\le \prod_{i=1}^n |R_i|^{1/(n-1)}=N^{1+1/(n-1)}.$
On the other hand, both Algorithm \ref{algo:LW} and
Algorithm \ref{algo:the-algo}
take $O(N^{1 + 1/(n-1)})$-time as we have analyzed.
In fact, for Algorithm \ref{algo:the-algo}, we are able to demonstrate
a stronger result: its run-time on this instance is $O(n^2N)$
which is better than what we can analyze for a general instance of this
type.
In particular, the run-time gap between Algorithm \ref{algo:the-algo}
and AGM's algorithm is $\Omega(N)$ for constant $n$.

\bp[Proof of Lemma \ref{LEM:BAD:INSTANCE}]
In the instances below the domain of any attribute will be
$\D = \{0,1,\dots,(N-1)/(n-1)\}$
For the sake of clarify, we  ignore the integrality issue.
For any $i\in [n]$, let $R_i$ be the set of {\em all} tuples
in $\D^{[n]-\{i\}}$ each of which has at most one non-zero value.
Then, it is not hard to see that $|R_i| = (n-1)[(N-1)/(n-1)+1] - (n-2) = N$,
for all $i\in [n]$; and,
$|\Join_{i=1}^n R_i| = n[(N-1)/(n-1)+1]-(n-1) = N+(N-1)/(n-1)>N$.

A relation $R$ on attribute set $\bar A \subseteq [n]$ is called ``simple"
if $R$ is the set of {\em all} tuples in $\D^{\bar A}$ each of which has at
most
one non-zero value. Then, we observe the following properties.
(a) The input relations $R_i$ are simple.
(b) An arbitrary projection of a simple relation is simple.
(c) Let $S$ and $T$ be any two simple relations on attribute sets
$\bar A_S$ and $\bar A_T$, respectively.
If $\bar A_S$ is contained in $\bar A_T$ or vice versa, then
$S \Join T$ is simple.
If neither $\bar A_S$  nor $\bar A_T$ is contained in the other,
then $|S\Join T| \geq (1+(N-1)/(n-1))^2 = \Omega(N^2/n^2)$.

For an arbitrary join-project plan starting from the simple relations
$R_i$,
we eventually must join two relations whose attribute sets are not
contained
in one another, which right then requires $\Omega(N^2/n^2)$ run time.
\ep

Finally, we analyze the run-time of Algorithm \ref{algo:the-algo}
directly on this instance without resorting to Lemma~\ref{lem:holder}.
H\"older's inequality lost some information about the run-time.
The following lemma shows that our algorithm and our bound can be
better than what we were able to analyze.

\begin{lmm}
On the collection of instances from the previous lemma,
Algorithm \ref{algo:the-algo} runs in time $O(n^2N)$.
\label{lmm:the-algo-on-bad-instance}
\end{lmm}
\bp
Without loss of generality, assume the hyperedge order
Algorithm~\ref{algo:the-algo} considers is
$[n]-\{1\}, \dots, [n]-{n}$.
In this case, the universe of the left-child of the root of the QP-tree
is $\{n\}$, and the universe of the right-child of the root is
$[n-1]$.

The first thing Algorithm~\ref{algo:the-algo} does is that it computes the join
$L_n = \Join_{i=1}^{n-1} \pi_{\{n\}}(R_i)$, in time $O(nN)$.
Note that $L_n = \D$, the domain.
Next, Algorithm~\ref{algo:the-algo} goes through each value
$a\in L_n$ and decide whether to solve a subproblem.
First, consider the case $a>0$.
Here Algorithm~\ref{algo:the-algo} estimates a bound for the join
$\Join_{j=1}^{n-1} \pi_{[n-1]}(R_j[a])$.
The estimate is $1$ because $|\pi_{[n-1]}(R_j[a])|=1$
for all $a>0$. Hence, the algorithm will recursively compute this join
which takes time $O(n^2)$ and filter the result against $R_n$.
Overall, solving the sub problems for $a>0$ takes $O(n^2N)$ time.
Second, consider the case when $a=0$. In this case
$|\pi_{[n-1]}(R_j[0])| = \frac{(n-2)N-1}{(n-1)}$.
The subproblem's estimated size bound is
\[ \prod_{i=1}^{n-1} |\pi_{[n-1]}(R_j[0])|^{\frac{1/(n-1)}{1-1/(n-1)}}
= \left[\frac{(n-2)N-1}{(n-1)}\right]^{(n-1)/(n-2)} > N
\]
if $N\geq 4$ and $n\geq 4$. Hence, in this case $R_n$ will be filtered
against
the $\pi_{[n-1]}(R_j[0])$, which takes $O(n^2N)$ time.
\ep

\paragraph*{Extending beyond LW instances} 
Using the above results, we give a sufficient condition
for when there exist a family of instances ${\cal I} = I_1,\dots,I_N,
\dots,$ such that on instance $I_N$ every binary join strategy takes time at
least
$\Omega(N^2)$, but our algorithm takes $o(N^2)$. Given a hypergraph
$H=(V,E)$. We first define some notation. Fix $U \subseteq V$ then
call an attribute $v \in V \setminus U$ {\em $U$-relevant} if for all $e$ such
that $v \in e$ then $e \cap U \neq \emptyset$; call $v$ 
{\em $U$-troublesome} if for all $e \in E$, if $v \in e$ then $U
\subseteq e$. Now we can state our result:

\begin{lmm}
Given a join query $H=(V,E)$ and some $U \subseteq V$ where $|U| \geq
2$, then if there exists $F \subseteq E$ such that $|F| = |U|$ that
satisfies the following three properties: (1) each $u \in U$ occurs in
exactly $|U|-1$ elements in $F$, (2) each $v \in V$ that is
$U$-relevant appears in at least $|U|-1$ edges in $F$, (3) there are no
$U$-troublesome attributes. Then, there is some family of
instances ${\cal I}$ such that (a) computing the join query
represented by $H$ with a join tree takes time $\Omega(N^2/|U|^2)$ while (b) the
algorithm from Section~\ref{sec:all:j} takes 
time $O(N^{1 + 1/(|U|-1)})$. 
\end{lmm}

Given a $(U,F)$ as in the lemma, the idea is to simply to set all
those edges in $f \in F$ to be the instances from
Lemma~\ref{LEM:BAD:INSTANCE} and extend all attributes with a single
value, say $c_0$. Since there are no $U$-troublesome attributes, to
construct the result set at least one of the relations in $F$ must be
joined. Since any pair $F$ must take time $\Omega(N^2/|U|^2)$ by the
above construction, this establishes (a). To establish (b), we need to
describe a particular feasible solution to the cover LP whose
objective value is $N^{1 + 1/(|U|-1)}$, implying that the running time
of our proposed algorithm is upper bounded by this value. To do
this, we first observe that any attribute not in $U$ takes the value
only $c_0$. Then, we observe that any node $v \in V$ that is not
$U$-relevant is covered by some edge $e$ whose size is exactly $1$
(and so we can set $x_e = 1$). Thus, we may assume that all nodes are
$U$-relevant. Then, observe that all relevant attributes can be set by
the cover $x_{e} = 1/(|U|-1)$ for $e \in F$. This is a feasible
solution to the LP and establishes our claim.

%% file: extensions.tex
\section{Extensions}
\label{sec:extensions}

In Section~\ref{sec:cc}, we describe some results on the combined
complexity of our approach.  Finally, in Section~\ref{sec:error}, we
observe that our algorithm can be used to compute a relaxed notion of
join.

\subsection{Combined Complexity} 
\label{sec:cc}\label{SEC:CC}
\newcommand{\sat}{\mathsf{3SAT}}
\newcommand{\usat}{\mathsf{3UniqueSAT}}

Given that our algorithms are data-optimal for 
worst-case inputs it is tempting to wonder if one can obtain an join 
algorithm whose run time is both query and data optimal in the worst-case.
We show that in the
special case when each input relation has arity at most $2$ we can 
attain a data-optimal algorithm that is simpler than Algorithm \ref{algo:the-algo}
with an asymptotically better query complexity.

Further, given promising results in the worst case, it is natural wonder if one can obtain a join algorithm whose run time is polynomial in both the size of the query \textit{as well} as the size 
of the output. More precisely, given a join query $q$ and an instance $I$, can 
one compute the result of query $q$ on instance $I$ in time 
$\poly(|q|,|q(I)|,|I|)$. Unfortunately, this is not possible unless $\np=\rp$.
We briefly present a proof of this fact below.

\paragraph*{Each relation has at most $2$ attributes}
As was mentioned in the introduction, our algorithm in 
Theorem~\ref{thm:main} not only has better data complexity than AGM's 
algorithm (in fact we showed our algorithm has optimal worst-case data 
complexity), it has a better query complexity. 
In this section, we show that for the special case when the join query 
$q$ is on relations with at most two attributes (i.e. the corresponding 
hypergraph $H$ is a graph), we can obtain an even better query complexity as 
in Theorem~\ref{thm:main} (with the same optimal data complexity).

Without loss of generality, we can assume that each relation contains
exactly $2$ attributes because a $1$-attribute relation $R_e$ needs to 
have $x_e=1$ in the corresponding LP and thus,
contributes a separate factor $N_e$ to the final product. 
Thus, $R_e$ can be joined with the rest of the query with any join algorithm 
(including the naive Cartesian product based algorithm).
In this case, the hypergraph $H$ is a graph which can be assumed to be
simple. 

We first prove an auxiliary lemma for the case when $H$ is a cycle.
We assume that all relations are indexed in advanced,
which takes $O(\sum_e N_e)$ time.
In what follows we will not include this preprocessing time in the analysis.
The following lemma essentially reduces the case when $H$ is a cycle to
the case when $H$ is a triangle, a Loomis-Whitney instance with $n=3$.

\begin{lmm}[Cycle Lemma]
If $H$ is a cycle, then $\Join_{e\in E} R_e$
can be computed in time $O(m\sqrt{\prod_{e\in H}N_e})$.
\label{lmm:cycle}
\end{lmm}
\bp
First suppose $H$ is an even cycle, consisting of consecutive
edges $e_1=(1,2)$, $e_2=(2,3)$,$\cdots$,$e_{2k'}=({2k'},1)$.
Without loss of generality, assume
\[ N_{e_1}N_{e_3} \cdots N_{e_{2k'-1}} \leq N_{e_2}N_{e_4} \cdots N_{e_{2k'}}. \]
In this case, we compute the (cross-product) join
\[ R = R_{e_1} \Join R_{e_3} \Join \cdots \Join R_{e_{2k'-1}}. \]
Note that $R$ contains all the attributes.
Then, sequentially join $R$ with each of $R_{e_2}$ to $R_{e_{2k'}}$.
The total running time is 
\[ O\left(k'N_{e_1}N_{e_3} \cdots N_{e_{2k'-1}}\right)
   = O\left(m\prod_{e\in H}N_e\right).
\]
Second, suppose $H$ is an odd cycle consisting of consecutive edges
$e_1=(1,2)$, $e_2=(2,3)$, $\dots$, $e_{2k'+1}=({2k'+1},1)$. 
If $k'=1$ then by the Loomis-Whitney algorithm for the 
$n=3$ case (Algorithm \ref{algo:LW}), 
we can compute $R_{e_1} \Join R_{e_2} \Join R_{e_3}$
in time $O(\sqrt{N_{e_1}N_{e_2}N_{e_3}})$.
Suppose $k'>1$. Without loss of generality, assume
\[ N_{e_1}N_{e_3} \cdots N_{e_{2k'-1}} \leq 
   N_{e_2}N_{e_4} \cdots N_{e_{2k'}}.
\]
In particular, 
$N_{e_1}N_{e_3} \cdots N_{e_{2k'-1}} \leq \sqrt{ \prod_{e\in H} N_e}$, which
means the following join can be computed in time
$O(m\sqrt{ \prod_{e\in H} N_e})$:
\[ X = R_{e_1} \Join R_{e_3} \Join \cdots \Join R_{e_{2k'-1}}. \]
Note that $X$ spans the attributes in the set $[2k']$.
Let $S=\{2,3,\dots,2k'-1\}$, and $X_S$ denote the projection
of $X$ down to coordinates in $S$; and define
\[ W = (\dots(X_S \Join R_{e_2}) \Join R_{e_4}) \cdots \Join R_{e_{2k'-2}}). \]
Since $R_{e_2} \Join R_{e_4} \cdots \Join R_{e_{2k'-2}}$ spans
precisely the attributes in $S$, the relation $W$ can be computed
in time $O(m|X_S|) = O(m|X|) = O(m\sqrt{ \prod_{e\in H} N_e})$.
Note that
\[ |W| \leq \min\{N_{e_1}N_{e_3}\cdots N_{e_{2k'-1}}, 
                  N_{e_2}N_{e_4}\cdots N_{e_{2k'-2}}\}. 
\]
We claim that one of the following inequalities must hold:
\begin{eqnarray*}
|W| \cdot N_{e_{2k'}} & \leq & \sqrt{ \prod_{e\in H} N_e}, \text{ or}\\
|W| \cdot N_{e_{2k'+1}} & \leq & \sqrt{ \prod_{e\in H} N_e}.
\end{eqnarray*}
Suppose both of them do not hold, then 
\begin{eqnarray*}
 \prod_{e\in H} N_e
&=&
 (N_{e_1}N_{e_3}\cdots N_{e_{2k'-1}}) \cdot (N_{e_2}N_{e_4}\cdots N_{e_{2k'-2}}) 
 \cdot N_{e_{2k'}} \cdot N_{e_{2k'+1}}\\
&\geq&|W|^2 N_{e_{2k'}} N_{e_{2k'+1}}\\
&=&(|W| \cdot N_{e_{2k'}}) \cdot (|W| \cdot N_{e_{2k'+1}})\\
&>& \prod_{e\in H} N_e,
\end{eqnarray*}
which is a contradiction.
Hence, without loss of generality we can assume
$|W| \cdot N_{2k'} \leq \sqrt{ \prod_{e\in H} N_e}$.
Now, compute the relation
\[ Y = W \Join R_{e_{2k'}}, \]
which spans the attributes $S\cup \{2k',2k'+1\}$.
Finally, by thinking of all attributes in the set $S\cup \{2k'\}$ as a 
``bundled attribute",
we can use the Loomis-Whitney algorithm for $n=3$ to compute the join
\[ X \Join Y \Join R_{e_{2k'+1}} \]
in time linear in
\begin{eqnarray*}
\sqrt{|X| \cdot |Y| \cdot N_{e_{2k'+1}}}
&\leq& \sqrt{(N_{e_1}N_{e_3}\cdots N_{e_{2k'-1}}) \cdot 
    (|W|\cdot N_{e_{2k'}}) \cdot N_{e_{2k'+1}}}\\
&\leq& \sqrt{(N_{e_1}N_{e_3}\cdots N_{e_{2k'-1}}) \cdot 
    (N_{e_2}N_{e_4}\cdots N_{e_{2k'-2}}\cdot N_{e_{2k'}}) \cdot N_{e_{2k'+1}}}\\
&=& \sqrt{ \prod_{e\in H} N_e}.
\end{eqnarray*}
\ep

With the help of Lemma \ref{lmm:cycle}, we can now derive a solution
for the case when $H$ is an arbitrary graph.
Consider any {\em basic feasible solution} $\mathbf x = (x_e)_{e\in E}$
of the fractional cover polyhedron 
\begin{eqnarray*}
\sum_{v\in e} x_e & \geq & 1, \ v \in V\\
x_e & \geq & 0, \ e \in E.
\end{eqnarray*}
It is known that $\mathbf x$ is {\em half-integral}, i.e.
$x_e \in \{0,1/2,1\}$ for all $e\in E$ 
(see Schrijver's book \cite{MR1956924}, Theorem 30.10).
However, we will also need a graph structure associated with
the half-integral solution; hence, we adapt a known proof 
\cite{MR1956924} of the half-integrality property
with a slightly more specific analysis.
It should be noted, however, that the following is already implicit in
the existing proof.

\begin{lmm}
For any basic feasible solution $\mathbf x = (x_e)_{e\in E}$
of the fractional cover polyhedron above, $x_e \in \{0,1/2,1\}$ for all
$e \in E$.
Furthermore, the collection of edges $e$ for 
which $x_e=1$ is a union $\mathcal S$
of stars. And, the collection of edges $e$ for which
$x_e=1/2$ form a set $\mathcal C$ of vertex-disjoint 
odd-length cycles that are also 
vertex disjoint from the union $S$ of stars.
\label{lmm:bfs-d=2}
\end{lmm}
\bp
First, if some $x_e= 0$, then we remove $e$ from the graph and recurse on 
$G-e$. 
The new $\mathbf x$ is still an extreme point of the new polyhedron.
So we can assume that $x_e > 0$ for all $e\in E$. 

Second, we can also assume that $H$ is connected. 
Otherwise, we consider each connected component separately.

Let $k=|V|$ and $m=|E|$.
The polyhedron is defined on $m$ variables and $k+m$ inequality constraints. 
The extreme point must be the intersection of exactly $m$ (linearly independent)
tight constraints. 
But the constraints $\mathbf x \geq \mathbf 0$ are not tight as
we have assumed $x_e>0, \forall e$. 
Hence, there must be $m$ vertices $v$ for which
the constraints $\sum_{v\in e} x_e \geq 1$ are tight.
In particular, $m\leq k$. Since $H$ is connected, it is either a tree, or 
has exactly one cycle.

Suppose $H$ is a tree, then it has at least $2$ leaves and at most one 
non-tight constraint (as there must be $m=k-1$ tight constraints). 
Consider the leaf $u$ whose constraint is tight. 
Let $v$ be $u$'s neighbor. Then $x_{uv} = 1$ because $u$ is tight. 
If $v$ is tight then we are done, the graph $H$ is just an edge $uv$. 
(If there was another edge $e$ incident to $v$ then $x_e=0$.)
If $v$ is not tight then $v$ is not a leaf. We start from another tight 
leaf $w \neq u$ of the tree and reason in the same way.
Then, $w$ has to be connected to $v$. Overall, the graph is a star.

Next, consider the case when $H$ is not a tree. 
All $k=m$ vertices has to be tight in this case. 
Thus, there cannot be a degree-$1$ vertex for the same reasoning as above. 
Thus, $H$ is a cycle. If $H$ is an odd cycle then it is easy to show that the 
only solution for which all vertices are tight is the all-$1/2$ solution. 
If $H$ is an even cycle then $\mathbf x$ cannot be an extreme point 
because it can be written as $\mathbf x = (\mathbf y+\mathbf z)/2$ for 
feasible solutions $\mathbf y$ and $\mathbf z$ 
(just add and subtract $\eps$ from alternate edges to form $\mathbf y$ and
$\mathbf z$).
\ep

Now, let $\mathbf x^*$ be an {\em optimal} basic feasible solution to the
following linear program.
\begin{eqnarray*}
\min & \sum_e (\log N_e) \cdot x_e\\
s.t. & \sum_{v\in e} x_e & \geq 1, \ v \in V\\
     & x_e & \geq 0, \ e \in E.
\end{eqnarray*}
Then $\prod_{e\in E}N_e^{x^*_e} \leq \prod_{e\in E}N_e^{x_e}$ for any
feasible fractional cover $\mathbf x$.
Let $S$ be the set of edges on the stars
and $\mathcal C$ be the collection of disjoint cycles as shown in the 
above lemma, applied to $\mathbf x^*$.
Then,
\[ \prod_{e\in E}N_e^{x^*_e}
   = \left(\prod_{e\in S}N_e\right)
      \prod_{C\in \mathcal C}\sqrt{\prod_{e\in C}N_e}.
\]
Consequently, we can apply Lemma \ref{lmm:cycle} to each cycle $C\in\mathcal C$
and take a cross product of all the resulting relations with the relations
$R_e$ for $e\in S$.
We just proved the following theorem.

\bthm
When each relation has at most two attributes, we can compute the join
$\Join_{e\in E}R_e$ in time $O(m\prod_{e\in E}N_e^{x_e})$.
\ethm

\paragraph*{Impossibility of Instance Optimality} 
The proof is fairly standard: we use the standard reduction of $\sat$ to 
conjunctive queries but with two simple specializations: (i) We reduce from 
the $\usat$, where the input formula is either unsatisfiable or has 
\textit{exactly} one satisfying assignment and (ii) $q$ is a full join 
query instead of a general conjunctive query. It is known that $\usat$ cannot 
be solved in deterministic polynomial time unless $\np=\rp$~\cite{unique-sat}.

For the sake of completeness, we sketch the reduction here. 
Let $\phi=C_1\wedge C_2\wedge \dots C_m$ be a $\usat$ CNF formula on $n$ 
variables $a_1,\dots,a_n$. (W.l.o.g. assume that a clause does not contain 
both a variable and its negation.)  For each clause $C_j$ for 
$j\in [m]$, create a relation $R_j$ on the variables that occur in 
$C_j$. The query $q$ is \[\Join_{j\in [m]} R_j.\]
Now define the database $I$ as follows: for each $j\in [m]$, $R_j^{I}$ 
contains the seven assignments to the variables in $C_j$ that makes it true. 
Note that $q(I)$ contains all the satisfying assignments for $\phi$: in 
other words, $q(I)$ has one element if $\phi$ is satisfiable otherwise 
$q(I)=\emptyset$. In  other words, we have $|q(I)|\le 1$, $|q|=O(m+n)$ and 
$|I|=O(m)$. Thus an instance optimal algorithm with time complexity 
$\poly(|q|,|q(I)|,|I|)$ for $q$ would be able to determine if $\phi$ is 
satisfiable or not in time $\poly(n,m)$, which would imply $\np=\rp$.

\input{error}

\subsection{Dealing with full queries and simple functional dependencies}
\label{app:sec:fqp}

\paragraph*{Full query processing}
Our goal in this section is to handle a more general class of queries
that may contain selections and joins to the same table, which we
describe now.

Our notation in this section follows Gottlob et
al's~\cite{GLVV09} notation, and we reproduce it
here for the sake of completeness. A {\em database instance} consists
$I = (\mathcal{U}, R_1,\dots,R_m)$ consists of a finite universe of
constants $\mathcal{U}$ and relations $R_1,\dots,R_m$ each over
$\mathcal{U}$. A {\em conjunctive query} has the form $q = R(x_0)
\leftarrow R_{i_1}(u_1) \wedge \dots \wedge R_{i_m}(u_m)$, where each
$u_j$ is a list of (not necessarily distinct) variables of length
$|u_j|$. We call each $R_{i_j}$ a subgoal. Each variable that occurs
in the head $R(u_0)$ must also appear in the body. We call a
conjunctive query {\em full} if each variable that appears in the body
also appears in the head. The set of all variables in $Q$ is denoted
$\mathrm{var}(Q)$. A single relation may occur several times in the
body, and so we may have $i_{j} = i_{k}$ for some $j \neq k$. The
answer of a query $q$ over a database instance $I$ is a set of tuples
of arity $|u_0|$, which is denoted $q(I)$, and is defined to contain
exactly those tuples $\theta(x_0)$ where $\theta : \mathrm{var}(Q) \to
\mathcal{U}$ is any substitution such that for each $j=1,\dots,m$,
$\theta(u_i) \in R_{i_j}$.

We call a full conjunctive query {\em reduced} if no variable is
repeated in the same subgoal. We can assume without loss of generality
that a full conjunctive query is reduced since we can create an
equivalent reduced query within the time bound. In time $O(|R_{i_j}|)$
for each $j=1,\dots,m$, we create a new relation $R_{i_j}'$ with arity
equal to the number of distinct variables. In one scan over $R_{i_j}$
we can produce $R_{i_j}'$ by keeping only those tuples that satisfy
constants (selections) in the query and any repeated variables. We
then construct $q'$ a query over the $R_{i_j}$ in the obvious
way. Clearly $q(I) = q'(I)$ and we can construct both in a single scan
over the input.  Finally, we make the observation that our method can
tolerate multisets as hypergraphs, and so our results extend our
method to full conjunctive queries. Summarizing our discussion, we
have a worst-case optimal instance for full conjunctive queries as
well.

\paragraph*{Simple Functional Dependencies}

Given a join query $(V,E)$, a (simple) functional dependency (FD) is a
triple $(e,u,v)$ where $u,v \in V$ and $e \in E$ and is written as $e.u
\to e.v$. It is a constraint in that the FD $(e,u,v)$ implies that for
any pair of tuples $\mv t,\mv t' \in R_e$,  if $t_u = t'_u$ then $t_v
= t'_v$. Fix a set of functional dependencies $\Gamma$, construct a
directed (multi-)graph $G(\Gamma)$ where the nodes are the attributes $V$ and 
there is an edge
$(u,v)$ for each functional dependency. The set of all nodes reachable
from a node $u$ is a set $U$ of nodes; this relationship is denoted $u
\to^{*} U$.

Given a set of functional dependencies, we propose an algorithm to
process a join query. The first step is to compute for each relation
$R_e$ for $e \in E$, a new relation $R'_{e'}$ whose attributes are the
union of the closure of each element of $v \in E$, i.e., $e' = \set{ u
  \suchthat v \to u \text{ for } v \in e}$. Using the closure this can be
computed in time $|E| |V|$. Then, we compute the contents of
$R_{e}'$. Walking the graph induced by the FDs in a breadth first
manner, we can expand $R_{e}$ to contain all the attributes $R_{e'}$
in time linear in the input size. Finally, we solve the LP from
previous section and use our algorithm. It is clear that this
algorithm is a strict improvement over our previous algorithm that is
FD-unaware. It is an open question to understand its data
optimality. We are, however, able to give an example that
suggests this algorithm can be substantially better than algorithms
that are not FD aware.

Consider the following family of instances on $k+2$ attributes
$A,B_1,\dots,B_k,C$ parameterized by $N$:
\[ q = \left(\Join_{i=1}^{k} R_{i}(A,B_i) \right) \Join
\left(\Join_{i=1}^{k} S_{i}(B_i, C) \right) \] 

Now we construct a family of instances such that $|R_i| = |S_i| = N$
for $i=1,\dots,k$. Suppose there are functional dependencies $A
\to B_i$. 

Our algorithm will first produce a relation $R'(A,B_1,\dots,B_k)$
which can then be joined in time $N$ with each relation $S_{i}$ for
$i=1,\dots,k$. When we solve the LP, we get a bound of of $|q(I)| \leq
N^2$ -- and our algorithm runs within this time.

Now consider the original instance without functional
dependencies. Then, the AGM bound is $|q(I)| \leq N^{k}$. More
interestingly, one can construct a simple instance where half of the
join has a huge size, that is $|\Join_{i=1}^{k} S_{i}(B_i, C)| =
N^{k}$. Thus, if we choose the wrong join ordering our algorithms
running time will blow up.

%% file: error.tex
\subsection{Relaxed Joins}
\label{sec:error}

\newcommand{\C}{\mathcal{C}}
\newcommand{\lpopt}{\mathsf{LPOpt}}
\newcommand{\bfs}{\mathsf{BFS}}

We observe that our algorithm can actually evaluate a relaxed notion of join 
queries. Say we are given a query $q$ represented by a hypergraph
$H=(V, E)$ where $V=[n]$ and $|E|=m$.
The $m$ input relations are $R_e$, $e\in E$.
We are also given a ``relaxation'' number $0\leq r \leq m$. 
Our goal is to output all tuples that agree 
with at least $m-r$ input relations. 
In other words, we want to compute 
$\cup_{S\subseteq E, |S|\ge m-r} \Join_{e\in S} R_e$. 
However, we need to modify the problem to avoid the case that the set of 
attributes of relations indexed by $S$ does not cover all the attributes 
in the universe $V$. 
Towards this end, define the set

\[ \C(q,r)= \left\{S\subseteq E \suchthat |S|\ge m-r\text{ and } 
   \bigcup_{e\in S} e = V \right\}.
\]

With the notations established above, we are now ready to define the 
relaxed join problem.

\bdefn[Relaxed join problem]
Given a query $q$ represented by the hypergraph $H=(V=[n],E)$,
and an integer $0\le r\le m$, evaluate
\[ q_r\stackrel{def}{=}\bigcup_{S\in \C(q,r)} \left(\Join_{e\in S} R_e\right).
\]
\edefn

Before we proceed, we first make the following simple observation: 
given any two sets $S, T\in \C(q,r)$ such that $S\subseteq T$, 
we have $\Join_{e\in T} R_e \subseteq \Join_{e\in S} R_e$. 
This means in the relaxed join problem we only need to consider 
subsets of relations that are not contained in any other subset. 
In particular, define 
$\hat{\C}(q,r)\subseteq \C(q,r)$ to be the largest subset of $\C(q,r)$
such that for any $S\neq T\in \hat{\C}(q,r)$ neither 
$S\subset T$ nor $T\subset S$. 
We only need to evaluate
$q_r = \bigcup_{S\in \hat{\C}(q,r)} \left(\Join_{e\in S} R_e\right).$

Given an $S\in\hat{\C}(q,r)$, let $\lpopt(S)$ denote the optimal size bound 
given by the AGM's fractional cover inequality \eqref{eqn:agm08-bound}
on the join query represented by the hypergraph $(V, S)$.
In particular, $\lpopt(S) = \prod_{e\in S}|R_e|^{x^*_e}$ 
where $\mv x^*_S = (x^*_e)_{e\in S}$ is an optimal solution to the
following linear program called LP$(S)$:
\begin{eqnarray}
\min & \sum_{e\in S} (\log |R_e|)\cdot x_e&\nonumber \\
\text{subject to}& \sum_{e\in S: i \in e} x_e \geq 1& \text{for any $i\in V$}
\label{eqn:LP(S)}\\
 & x_e \geq 0 & \text{for any } e \in S.\nonumber
\end{eqnarray}

\paragraph*{Upper bounds} We start with a straightforward upper bound. 
\begin{prop}
\label{prop:error-bound-loose}
Let $q$ be a join query on $m$ relations and let $0\le r\le m$ be an integer. 
Then given sizes of the input relations, the number of 
output tuples for query $q_r$ is upper bounded by
\[\sum_{S\in \hat{\C}(q,r)} \lpopt(S).\]
\end{prop}

Further, Algorithm~\ref{algo:the-algo} evaluates $q_r$ with data 
complexity linear in the bound above. 
The next natural question is to determine how good the upper bound is. 
Before we answer the question, we prove a stronger upper bound.

Given a subset of hyperedges $S\subseteq E$ which ``covers'' $V$,
i.e. $\cup_{e\in S} e = V$,  let $\bfs(S)\subseteq S$ be the subset of 
hyperedges in $S$ that gets a {\em positive} $x^*_e$ value in an {\em optimal} 
basic feasible solution to the linear program LP$(S)$ defined
in \eqref{eqn:LP(S)}.
(If there are multiple such solutions, pick any one in a consistent manner.) 
Call two subsets $S,T\subseteq E$ \textit{bfs-equivalent} if 
$\bfs(S)=\bfs(T)$. 
Finally, define $\C^*(q,r)\subseteq \hat{C}(q,r)$ as the collection of 
sets from $\hat{\C}(q,r)$ which contains exactly one arbitrary 
representative from each bfs-equivalence class.

\begin{thm} 
\label{thm:error-bound-tight}
Let $q$ be a join query represented by $H=(V,E)$, and let $0\le r\le m$ be 
an integer. 
The number of output tuples of $q_r$ is upper bounded by
$\sum_{S\in \C^*(q,r)} \lpopt(S).$
Further, the query $q_r$ can be evaluated in time 
\[ O\left(\sum_{S\in \C^*(q,r)} \left(mn\cdot \lpopt(S)+\poly(n,m)
   \right)\right) \]
plus the time needed to compute $\C^*(q,r)$ from $q$.
\end{thm}

Note that since $\C^*(q,r)\subseteq \hat{\C}(q,r)$, the bound in 
Theorem~\ref{thm:error-bound-tight} is no worse than that in 
Proposition~\ref{prop:error-bound-loose}. We will show later that the bound 
in Theorem~\ref{thm:error-bound-tight} is indeed tight.

\begin{proof}[Proof of Theorem~\ref{thm:error-bound-tight}] 
We will prove the result by presenting the algorithm to compute $q_r$. 
A simple yet key idea is the following.
Let $S \neq S' \in \hat{C}(q,r)$ be two different sets of hyperedges with the following property.
Define $T\stackrel{def}{=}\bfs(S)=\bfs(S')$ and let 
$\mv x^*_{T}=(x^*_i)_{i\in T}$ be the projection of the corresponding
optimal basic feasible solution to the $(V,S)$ and the $(V,S')$ problems
projected down to $T$. (The two projections result in the same
vector $\mv x^*_T$.)
The outputs of the joins on $S$ and on $S'$ are both subsets of the
output of the join on $T$. We can simply run
Algorithm~\ref{algo:the-algo} on inputs $(V,T)$ and $\mv x^*_T$,
then prune the output against relations $R_e$ with $e\in S\setminus T$ or $S'\setminus T$.
In particular, we only need to 
compute $\Join_{e\in T} R_e$ once for both $S$ and $S'$.

\floatname{algorithm}{Algorithm}
\begin{algorithm}
\caption{Computing Relaxed Join $q_r$}
\label{algo:relaxed-join}
\begin{algorithmic}[1]
\STATE Compute $\C^*(q,r)$.
\STATE $Q\leftarrow \emptyset$.
\FOR{ every $S\in \C^*(q,r)$}
   \STATE Let $\mv x^*_S$ be an optimal BFS for LP$(S)$
   \STATE Let $T=\{e\in S \suchthat x^*_e>0\}$. (Note that $T=\bfs(S)$.)
   \STATE Run Algorithm~\ref{algo:the-algo} on $\{x^*_e\}_{e\in T}$ to compute
          $\phi_{T}=\Join_{e\in T} R_e$.
   \FOR{ every tuple $\mv t\in \phi_{T}$}
      \IF{for at least $m-r$ hyperedges $e\in E$, $\mv t_e \in R_e$} 
          \STATE $Q \la Q \cup \{\mv t\}$
      \ENDIF
   \ENDFOR
\ENDFOR
\RETURN $Q$
\end{algorithmic}
\end{algorithm}

Other than the time to compute $\C^*(q,r)$ in the line 1, line 4 needs 
$\poly(n,m)$ time to solve the LP, line 5 needs $O(m)$ time, while by 
Theorem~\ref{thm:main}, line 6 will take 
$O(mn\cdot\lpopt(S)+m^2n)$ time. Finally, 
Theorem~\ref{thm:main} shows that $|\phi_{T}|\le \lpopt(S)$,\footnote{This also proves the claimed bound on the size of $q_r$.} which shows that the 
loop in line 7 is repeated $\lpopt(S)$ times and lines 8-9 can be 
implemented in $O(m)$ time and thus, lines 7-9 will take time 
$O(m\cdot \lpopt(S))$.

Finally, we argue the correctness of the algorithm. We first note that by 
line 8, every tuple $\mv t$ that is output is indeed a correct one. 
Thus, we have to argue that we do not miss any tuple $\mv t$ that needs to be 
output. For the sake of contradiction assume that there exists such a 
tuple $\mv t$. Note that by definition of $\hat{\C}(q,r)$, this implies that 
there exists a set $S'\in \hat{\C}(q,r)$ such that for every 
$e\in S'$, $\mv t_e\in R_e$. However, note that by definition of 
$\C^*(q,r)$, for some execution of the loop in line 3, we will consider 
$T$ such that $T=\bfs(S')$. Further, by the correctness of 
Algorithm~\ref{algo:the-algo}, we have that $\mv t\in \phi_T$. 
This implies (along with the definition of $\hat{\C}(q,r)$) that 
$\mv t$ will be retained in line 8, which is a contradiction.
\end{proof}

It is easy to check that one can compute $\C^*$ in time $m^{O(r)}$ 
(by going through all subsets of $E$ of size at least $m-r$ and 
performing all the required checks). We leave open the question of 
whether this time bound can be improved.

\paragraph*{Lower bound} We now show that the bound in 
Theorem~\ref{thm:error-bound-tight} is tight for some query
and some database instance $I$.

We first define the query $q$. The hypergraph is $H=(V=[n],E)$
where $m=|E|=n+1$. The hyperedges are $E=\{e_1,\dots,e_{n+1}\}$
where $e_i=\{i\}$ for $i\in [n]$ and
$e_{n+1} = [n]$.
The database instance $I$ consists of
relations $R_e$, $e\in E$, all of which are of size $N$.
For each $i\in [n]$, $R_{e_i} = [N]$.
And, $R_{e_{n+1}} = \bigcup_{i=1}^N \{N+i\}^n$.


It is easy to check that for any $r>0$, $q_r(I)$ is the set 
$R_{e_{n+1}} \cup  [N]^n$, i.e.  
$|q_r(I)| = N+N^n.$
Next, we claim that for this query instance 
$\C^*(q,r)=\{\{n+1\}, [n]\}$. Note that $\bfs(\{n+1\})=\{n+1\}$ and 
$\bfs([n])=[n]$, which implies that $\lpopt(\{n+1\})=N$ and 
$\lpopt([n])=N^n$. This along with Theorem~\ref{thm:error-bound-tight} implies 
that $|q_r(I)|\le N+N^n$, which proves the tightness of the size bound in 
Theorem~\ref{thm:error-bound-tight}, as desired.

Finally, we argue that $\C^*(q,r)=\{ \{n+1\}, [n]\}$. Towards this end, 
consider any $T\in \hat{\C}(q,r)$. Note that if $(n+1)\not\in T$, we have 
$T=[n]$ and since $\bfs(T)=T$ (and we will see soon that for any other 
$T\in\hat{\C}(q,r)$, we have $\bfs(T)\neq [n]$), which implies that 
$[n]\in\C^*(q,r)$. Now consider the case when $(n+1)\in T$. Note that in this 
case $T=\{n+1\} \cup T'$ for some $T'\subset [n]$ such that $|T'|\ge n-r$. 
Now note that all the relations in $T$ cannot cover the $n$ attributes but 
$R_{n+1}$ by itself does include all the $n$ attributes. This implies that 
$\bfs(T)=\{n+1\}$ in this case. This proves that $\{n+1\}$ is the other 
element in $\C^*(q,r)$, as desired.

Finally, if one wants a more general example where $m=n+k$ for $k>1$, then 
one can repeat the above instance $k$ times, where each repetition has $n/k$ 
fresh attributes. In this case, $\C^*$ will consists of all subsets of 
relation where in each repetition, each such subset has exactly one of 
$\{n/k+1\}$ or $[n/k]$. In particular, the query output size will be 
$\sum_{i=0}^r \binom{k}{i}\cdot N^{k-i}\cdot N^{n\cdot i/k}$.

%% file: conclusions.tex
\section{Conclusion and Future Work}
\label{sec:conc}

In this work, we established optimal algorithms for the worst-case
behavior of join algorithms. We also demonstrated that the join
algorithms employed in RDBMSes do not achieve these optimal bounds --
and we demonstrated families of instances where they were
asymptotically worse by factors close to the size of the largest
relation. It is interesting to ask similar questions for average case
complexity. Our work offers a fundamentally different way to approach
join optimization rather than the traditional
binary-join/dynamic-programming-based approach. Thus, our immediate
future work is to implement these ideas to see how they compare in
real RDBMS settings to the algorithms in a modern RDBMS.

Another interesting direction is to extend these results to a larger
classes of queries and to database schemata that have constraints. We
include in the appendix some preliminary results on full conjunctive
queries and simple functional dependencies (FDs). Not surprisingly,
using dependency information one can obtain tighter bounds
compared to the (FD-unaware) fractional cover technique.  
We will also investigate whether our algorithm for computing
relaxed joins can be useful in related context such as those considered
in Koudas et al~\cite{Koudas06relaxingjoin}.

There are potentially interesting connections between our work and 
several inter-related topics, which are all great subjects to further
explore. We algorithmically proved AGM's bound
which is equivalent to BT inequality, which in turn is essentially
equivalent to Shearer's entropy inequality. There are known combinatorial
interpretations of entropy inequalities which Shearer's is a special
case of; for example, Alon et al. \cite{DBLP:journals/ejc/AlonNSTV07} 
derived some such connections
using a notion of ``sections" similar to what we used in this paper.
An analogous partitioning procedure was used in 
\cite{DBLP:conf/stoc/Marx10} to compute joins by relating the number
of solutions to submodular functions.
Our lead example (the LW inequality with $n=3$) is equivalent to the problem of
enumerating all triangles in a tri-partite graph. It was known that this can be
done in time $O(N^{3/2})$ \cite{DBLP:journals/algorithmica/AlonYZ97}.

\paragraph*{Acknowledgments} 
We thank Georg Gottlob for sending us a full version of his
work~\cite{GLVV09}. We thank XuanLong Nguyen for introducing us to the
Loomis-Whitney inequality. We thank the anonymous referees for many helpful
comments which have greatly improved the presentation clarity.
CR's work on this project is generously
supported the NSF CAREER Award under IIS-1054009, the Office of Naval
Research under award N000141210041, and gifts or research awards from
Google, Greenplum, Johnson Controls, LogicBlox, and Oracle.